 \documentclass[draft]{article}

\usepackage{amsmath,amsfonts,amsthm,amssymb,amscd,cancel,color}
\usepackage{enumitem}
\usepackage{verbatim}
\usepackage[dvipdfmx]{graphicx}
\usepackage{ulem,color}

\renewcommand{\Re}{\mathop{\rm Re}\nolimits}
\renewcommand{\Im}{\mathop{\rm Im}\nolimits}

\theoremstyle{plain}
\newtheorem{theorem}{Theorem}[section]
\newtheorem{lemma}[theorem]{Lemma}
\newtheorem{proposition}[theorem]{Proposition}
\newtheorem{corollary}[theorem]{Corollary}
\theoremstyle{definition}
\newtheorem{definition}[theorem]{Definition}
\theoremstyle{remark}
\newtheorem{remark}[theorem]{Remark}

\newtheorem{assumption}[theorem]{Assumption}

\newcommand{\R}{{\mathbb R}}

\newcommand{\Z}{{\mathbb Z}}

\newcommand{\N}{{\mathbb N}}

\def\im{{\rm i}}

\newcommand{\C}{\mathbb{C}}
\newcommand{\T}{\mathbb{T}}

\def\({\left(}
\def\){\right)}
\def\<{\left\langle}
\def\>{\right\rangle}

\numberwithin{equation}{section}

\begin{document}

\title{Asymptotic stability of small bound state of\\ nonlinear quantum walks}

\author{Masaya Maeda}
\maketitle

\begin{abstract}
In this paper, we study the long time behavior of nonlinear quantum walks when the initial data is small in $l^2$.
In particular, we study the case where the linear part of the quantum walk evolution operator has exactly two eigenvalues and show that the solution decomposed into nonlinear bound states bifurcating from the eigenvalues and scattering waves.
\end{abstract}


\section{Introduction}
Discrete time quantum walks, which we refer quantum walks (QWs) in this paper, are space-time discrete dynamics which can be considered as a quantum counter part of classical random walks \cite{ABNVW01}.
Quantum mechanics is linear, so at first glance it seems difficult to consider nonlinear QWs.
However, it is possible to realize QWs with classical waves, such as light using optical devices \cite{BMKSW99PRA,KRS03OC}, and by this reason it becomes possible to implement nonlinearity in QWs.
In fact, models of nonlinear QWs are proposed using nonlinear optical elements \cite{NPR07PRA} and feed-forward mechanisms \cite{SWH14SR}.
In addition to experimental realizations, nonlinear QWs can be used to simulate nonlinear Dirac equations \cite{LKN15PRA, MS20RMP} and computing long time dynamics of wave packet spreading in nonlinear distorted media \cite{VFF19PRL}.

As Dirac equations \cite{EG21JMPA} (which are the continuous limit of QWs \cite{ANF14JPA,MA20QIP,MS20RMP}) and discrete Schr\"odinger equations \cite{Cuccagna09JMAA,EKT15JST,KKK06AA,PS08JMP} (also called continuous time QWs), QWs are dispersive \cite{MSSSSdis,MSSSS18DCDS,ST12JFA} in the sense that the $l^\infty$ norm of the solution decreases while $l^2$ norm of the solution conserves.
Given a nonlinear dispersive equation, it is natural to ask the existence and stability of bound state (also called soliton) solutions.
Recall that for KdV equation, which is one of the most famous nonlinear dispersive equation, solitons were found even before KdV equation itself \cite{Russell1844}.
The importance of solitons can also be seen by soliton resolution conjecture \cite{Soffer06ICM}, which claims generic solutions decouple into scattering waves and solitons, as in RAGE theorem for the linear case \cite{FO17JFA, SS16QSMF}.
Indeed, many papers studying nonlinear QWs numerically observe solitonic behavior of the solution and focus on the study of its dynamics \cite{BD2010.15281v1,BD20PRA,MDB15PRE,LKN15PRA,MSSSS19JPC,NPR07PRA,VFZF18Chaos}.
For the stability analysis of bound states, related to the study of topological phases \cite{Asboth12PRB,AO13PRB,CGSVWW18AHP,CGSVWW16JPA,Kitagawa12QIP,KRBD10PRA,Matsuzawa20QIP,ST19QIP,Suzuki19QIP,TAD14PRA}, Gerasimenko, Tarasinski, and Beenakker \cite{GTB16PRA}, followed by Mochizuki, Kawakami and Obuse \cite{MKO20JPA} studied the linear stability of bound states bifurcating from linear bound states.

Motivated by the above results, the main results of this paper is the following:
\begin{itemize}
\item Construct nonlinear bound states bifurcating from linear bound states (Proposition \ref{prop:nlbs}).
\item 
Prove the nonlinear (dynamical) stability of nonlinear bound states which we have constructed (Theorem \ref{thm:main}).
\end{itemize}
Actually, we will prove stronger result in Theorem \ref{thm:main}.
That is, we will show the soliton resolution for small solutions (i.e.\ we show all small solutions decouple into a nonlinear bound state and scattering waves).
Such asymptotic stability results have been investigated for nonlinear bound states for nonlinear Schr\"odinger equations \cite{GNT04IMRN,KM09JFA,Mizumachi08JMKU,SW90CMP}, nonlinear Dirac equations \cite{Boussaid06CMP,Boussaid08SIMA,PS12JMP} and discrete nonlinear Schr\"odinger equations \cite{CT09SIMA,KPS09SIAM,Maeda17SIMA,MP12DCDS}.
However, for discrete dynamical systems having dispersive properties, including QWs, we are not aware of such results.


Before going into the details we summarize notations which we use though out the paper:

\subsection{notations}
\begin{itemize}
	\item We write $a\lesssim b$ if there exists a constant $C$ s.t.\ $a\leq Cb$.
	Further, if the constant $C$ depends on some parameter $s$, we write $a\lesssim_s b$.
	\item
	Let $X$ be a Banach space, $x\in X$ and $r>0$, we set $B_X(x,r):=\{y\in X\ |\ \|x-y\|_X<r\}$ and $\overline{B_X(x,r)}:=\{y\in X\ |\ \|x-y\|_X \leq r\}$.
	\item For Banach spaces $X,Y$, we set $\mathcal{L}(X,Y)$ to be the Banach space of bounded operators from $X$ to $Y$ equipped with the usual operator norm $\|T\|_{\mathcal{L}(X,Y)}=\sup_{\|u\|_X = 1}\|Tu\|_Y$.
	Further, we set $\mathcal{L}(X):=\mathcal{L}(X,X)$.
	Notice that $\mathcal{L}(\C^n)$ are merely set of $n\times n$ matrices.
	\item For an operator $U$ on a Hilbert space, $\sigma(U)$ (resp.\ $\sigma_{\mathrm{d}}(U)$) will denote the set of spectrum (discrete spectrum) of $U$.
	\item For a Banach space $X$, $p\in[1,\infty]$ and $s\in\R$, we set 
		\begin{align*}
l^{p,s}(\Z,X):=\{u:\Z\to X\ |\ \|u\|_{l^{p,s}(\Z,X)}^p:=\sum_{x\in \Z}\<x\>^{ps}\|u(x)\|_X^p <\infty\},
		\end{align*}
		where $\<x\>:=(1+|x|^2)^{1/2}$.
		Moreover, we set $l^p(\Z,X):=l^{p,0}(\Z,X)$.
		\item For $F \in C^1(X,Y)$, the Fr\'echet derivative of $F$ will be denoted by $DF$.
			In particular, for $x,w\in X$, $DF(x)w=\left.\frac{d}{d\epsilon}\right|_{\epsilon=0}F(x+\epsilon w)$.
	\item For $u={}^t(u_{\uparrow}\ u_{\downarrow}),v={}^t(v_{\uparrow}\ v_{\downarrow})\in \C^2$, we set $(u,v)_{\C^2}:=u_\uparrow\overline{v_\uparrow}+u_{\downarrow}\overline{v_\downarrow}$ and $\<u,v\>_{\C^2}:=\mathrm{Re}(u,v)_{\C^2}$.
	\item For $u,v\in l^2(\Z,\C^2)$, we set $(u,v):=\sum_{x\in \Z}(u(x),v(x))_{\C^2}$ and $\<u,v\>:=\Re (u,v)$.

\end{itemize}
\subsection{Set up and main results}
We define the shift operator $S\in \mathcal{L}(l^2(\Z,\C^2))$ by
\begin{align*}
	Su(x):=\begin{pmatrix}
		u_{\uparrow}(x-1)\\ u_{\downarrow}(x+1)
	\end{pmatrix},\ \text{where} \ u(x)=\begin{pmatrix}
			u_{\uparrow}(x)\\ u_{\downarrow}(x)
		\end{pmatrix}.
\end{align*}
Let $C:\Z\to U(2)$, where $U(2)$ is the set of $2\times 2$ unitary matrices.
A coin operator $\hat{C}\in \mathcal{L}(l^2(\Z,\C^2))$ associated to such map is defined by
\begin{align*}
	\hat{C}u(x):=C(x)u(x).
\end{align*}
By abuse of notation we identify $C$ and $\hat{C}$ and simply write $\hat{C} =C $.
Recall that general elements of $ U(2)$ can be expressed as
\begin{align*}
	e^{\im \theta}\begin{pmatrix}
		\beta & \overline{\alpha}\\ -\alpha & \overline{\beta}
	\end{pmatrix},\ \theta\in\R,\ \alpha,\beta\in\C\text{ and } \ |\alpha|^2+|\beta|^2=1.
\end{align*} 
We assume:
\begin{assumption}\label{ass:1}
	There exists  $C_\infty=\begin{pmatrix}
		\beta_\infty & \overline{\alpha_\infty}\\ -\alpha_\infty & \overline{\beta_\infty}
	\end{pmatrix}\in U(2)
	$ satisfying $0<|\alpha_\infty|<1$
	s.t. $$\|C(\cdot)-C_\infty\|_{l^{1,1}(\Z,\mathcal{L}(\C^2))}<\infty.$$
\end{assumption} 

\begin{remark}
	There is no loss of generality assuming $\theta_\infty :=\lim_{|x|\to \infty }\theta(x)=0$.
\end{remark}

The linear time evolution operator $U\in \mathcal{L}(l^2(\Z,\C^2))$ of quantum walk is given by
\begin{align*}
	U=SC.
\end{align*} 
Similarly, we set $U_\infty:=SC_\infty$.
We note that all $S,C,C_\infty,U,U_\infty$ are unitary operators.

We next introduce the nonlinear coin operator.
Let $\gamma \in \mathcal{L}(\C^2)$ be self-adjoint and $g\in C^\infty(\R,\R)$ with $g(0)=0$.
Then, we set $N=N_{\gamma,g}:\C^2\to \C^2$ by
\begin{align}\label{eq:N}
	N(u):=e^{\im g(\<u,\gamma u\>_{\C^2})\gamma}u.
\end{align}
As the coin operator, $N$ will be used to express the (nonlinear) operator on $l^2(\Z,\C^2)$.
The nonlinear coin operator given in the above form covers many nonlinear QWs appeared in the literature such as \cite{LKN15PRA,NPR07PRA}.
\begin{remark}
The nonlinearity considered in \cite{NPR07PRA} actually have the form
\begin{align*}
N(u)=e^{\im g_1(\<u,\gamma_1 u\>)\gamma_1} e^{\im g_2(\<u,\gamma_2 u\>)\gamma_2}u,\ \gamma_1=\begin{pmatrix}
1 & 0\\ 0 & 0
\end{pmatrix},\ 
\gamma_2=\begin{pmatrix}
0 & 0\\ 0 & 1
\end{pmatrix},
\end{align*}
with $g_1(s)=g_2(s)=s$,
and so it can be considered as having two nonlinear coins.
One can handle such nonlinearity by trivial modification.
\end{remark}
The specific choice of the nonlinearity \eqref{eq:N} is also natural in the sense that when we consider the continuous limit of nonlinear QWs, the limit becomes nonlinear Dirac equation having Hamiltonian, see e.g. \cite{MS20RMP}.
The above specific form is used in Lemma \ref{lem:symporth} to obtain the "orthogonality" of the linearized operator (see, Lemma \ref{lem:symporth} below).
An important example which do not fall in the above framework is the model considered in \cite{GTB16PRA,MKO20JPA}, with the nonlinear coin given by
\begin{align*}
N(u)=e^{\im g(\<u,\sigma_3 u\>)\sigma_2}u,
\end{align*}
where $\sigma_j$'s are the Pauli matrices.
Study of such wider class of nonlinearity may be a good direction for the future research.
 
\begin{remark}
The assumption $g(0)=0$ is not essential.
For general $g$, we can write $CN(u)=Ce^{g(0)\gamma} Ce^{\im \(g(\<u,\gamma u\>_\C^2)-g(0)\)\gamma}u$.
Thus, replacing $C$ by $Ce^{g(0)\gamma}$ and $g$ by $g-g(0)$, we can reduce the general case to the case $g(0)=0$.
\end{remark}
The time evolution of nonlinear QW is given by the following recursion relation:
\begin{align}\label{eq:nlqw}
	u(t+1)=UN(u(t)).
\end{align}
In other words, if $u\in l^\infty(\Z,l^2(\Z,\C^2))$ satisfies \eqref{eq:nlqw}, we say $u$ is a solution of nonlinear QW \eqref{eq:nlqw}.
Notice that we have $\|u(t)\|_{l^2}=\|u(0)\|_{l^2}$ for all $t\in \Z$.
\begin{remark}
In the context of quantum mechanics, $\|u\|_{l^2}^2$ is the total probability so it is natural to consider the case $\|u\|_{l^2}=1$ and such restriction makes no difference in the linear case.
However, in this paper we mainly consider the case $\|u\|_{l^2}\ll 1$, which is a natural starting point when we consider nonlinear problems.
We note that considering the nonlinear QW with the coin
\begin{align*}
N(u)=e^{\im g(\chi\<u,\gamma u\>)\gamma}u,
\end{align*}
and $\chi \ll 1$ with $\|u\|_{l^2}=1$ is equivalent to considering the same coin with $\chi=1$ and $\|u\|_{l^2}\ll1$.
Such tuning parameter $\chi$ is widely used in the study of nonlinear QWs and our result can be considered as a study of dynamics in the weak nonlinear regime.
We will study the nonlinear QWs in strong nonlinear regime ($\|u\|_{l^2}$ and $\chi$ are not small) in the forthcoming paper.
\end{remark}

Our first result is the existence of bound states bifurcating from $\sigma_{\mathrm{d}}(U)$ (the discrete spectrum of $U$).

\begin{proposition}[Nonlinear bound states]\label{prop:nlbs}
	Assume $e^{\im \lambda} \in \sigma_{\mathrm{d}}(U)$ and let $\phi$ be the associated eigenfunction.
	Then, for arbitrary $s>0$ there exists $\delta_s>0$ s.t.\ there exists $\Lambda[\cdot] \in C^\infty(B_{\C}(0,\delta_s),\R)$ and $\Phi[\cdot]\in C^\infty(B_{\C}(0,\delta_s),l^{2,s}(\Z,\C^2))$ s.t. for $z\in B_{\C}(0,\delta_s)$, 
	\begin{align}\label{nlbs:1}
		UN(\Phi[z])=e^{\im \Lambda[z]}\Phi[z].
	\end{align}
	Moreover, for $z\in B_{\C}(0,\delta_s)$ we have 
	\begin{align}\label{nlbs:2}
		\Phi[e^{\im \theta}z]=e^{\im \theta}\Phi[z]\ \text{ and }\ \Lambda[e^{\im\theta}z]=\Lambda[z],
	\end{align}
and for $w\in \C$, we have
\begin{align}\label{nlbs:3}
	\|\Phi[z] - z\phi\|_{l^{2,s}}\lesssim_s |z|^3,\quad \|D\Phi[z]w-w\phi\|_{l^{2,s}}\lesssim_s |z|^2|w|,
\end{align}
\end{proposition}

\begin{remark}
It is known that all $e^{\im \lambda}\in \sigma_{\mathrm{d}}$ are simple eigenvalues \cite{MSSSSdis}.
\end{remark}

\begin{remark}
Notice that for given $z\in B_{\C}(0,\delta_0)$, $u(t)=e^{\im \Lambda[z]t}\Phi[z]$ is a solution of nonlinear QW \eqref{eq:nlqw}.
\end{remark}

\begin{remark}
The above result is not restricted for the nonlinearity given in \eqref{eq:N}.
Following the argument of the proof, it is easy to prove Proposition \ref{prop:nlbs} for smooth local nonlinearity satisfying the gauge symmetry $e^{\im\theta}N(u)=N(e^{\im \theta}u)$.
\end{remark}

Before going to the explanation of the main theorem, we would like to introduce a symmetry of QWs.
Set the "zig-zag" transform $Z\in \mathcal{L}(l^2(\Z,\C^2))$ by
\begin{align*}
	Zu(x)=(-1)^xu(x).
\end{align*}
It is easy to show
\begin{align}\label{eq:anticomUZ}
	UZ=-ZU.
\end{align}
This symmetry is called Chiral symmetry in \cite{SK10PRE,Kitagawa12QIP}, sublattice symmetry in \cite{Kitagawa12QIP} and stroboscopic sublattice factorization in \cite{VFZF18Chaos}.
By \eqref{eq:anticomUZ}, we conclude:
\begin{proposition}
	$\sigma(U)=-\sigma(U)$ and $\sigma_{\mathrm{d}}(U)=-\sigma_{\mathrm{d}}(U)$.
	Moreover, if $\phi$ is the eigenfunction of $U$ associated to $e^{\im \lambda}\in \sigma_{\mathrm{d}}(U)$, $Z\phi$ is the eigenfunction of $U$ associated to $-e^{\im \lambda}$.
\end{proposition}
Since $Z^2=1$, and $Z$ is self-adjoint,
\begin{align*}
	P_{\pm}=\frac{1}{2}\(1\pm Z\),
\end{align*}
are orthogonal projections satisfying $P_++P_-=1$, $P_+P_-=0$.
Further, we have
\begin{align*}
	UP_\pm =P_\mp U\text{ and }\ N(P_\pm \cdot) = P_\pm N(\cdot).
\end{align*}
Thus, letting $u(t;u_0)$ be the solution of nonlinear QW \eqref{eq:nlqw} with $u(0)=u_0$, we have
\begin{align*}
	u(t;u_0) = u(t;P_+u_0)+u(t;P_-u_0),
\end{align*}
with $u(t;P_\pm u_0)=P_\pm u(t;P_\pm u_0)$ if $t$ is even and $u(t;P_\pm u_0)=P_\mp u(t;P_\pm u_0)$ if $t$ is odd.
This implies that nonlinear QW \eqref{eq:nlqw} consists of two non-interacting dynamics.
In the following, we will only consider solutions with $u(0)=P_+u(0)$.
Further, if we understand the dynamics of $u(t)$ with $t\in 2\Z$, then we can compute $u(t+1)$ easily by \eqref{eq:nlqw}.
Therefore, we will only consider $t\in 2\Z$.
Having this in mind, we define
\begin{align*}
	\mathcal{U}(u):=UN(UN(u)),
\end{align*}
and by retaking the time, we consider the following evolution equation
\begin{align}\label{eq:nlqw2}
	u(t+1)=\mathcal{U}(u(t)),\ u(0)=u_0\in l_+^2(\Z,\C^2):=P_+l^2(\Z,\C^2),
\end{align}
instead of \eqref{eq:nlqw}.
Notice that $\Phi_+[z]:=P_+\Phi[z]$ satisfies
\begin{align}\label{eq:Phiplus}
	\mathcal{U}(\Phi_+[z])=e^{\im \Lambda_+[z]}\Phi_+[z],
\end{align}
where $\Lambda_+[z]:=2\Lambda[z]$ so $u(t)=e^{\im \Lambda_+[z]t}\Phi_+[z]$ is a solution of \eqref{eq:nlqw2} for $z\in B_{\C}(0,\delta_0)$.

For the main result.
We need additional assumptions for the linear evolution operator $U$ and the nonlinear term $N$.
First, for the nonlinear term, we assume:
\begin{assumption}\label{ass:nonlinear}
$g'(0)=g''(0)=0$.
\end{assumption}
Thus, the typical nonlinearity we have in mind is $g(s)=s^3$.

For the spectrum of $U$, we assume:
\begin{assumption}\label{ass:linear}
	$U$ is generic and have exactly two discrete spectrum.
\end{assumption}
An explicit example we have in mind is given in \cite{KLS13QIP} where the coin is perturbed only at the origin.
We set $\pm e^{\im \lambda}$ be the two discrete spectrum of $U$ with eigenvectors $\phi$ and $Z\phi$.
Then, $e^{2\im \lambda}$ will be the unique discrete spectrum of $U^2$ restricted on $l_+^2(\Z,\C^2)$.
We normalize the corresponding eigenvector $\phi_+=P_+\phi$ so that $\|\phi_+\|_{l^2}=1$.
We set 
\begin{align*}
	P_c=1-(\cdot,\phi_+)\phi_+.
\end{align*}

We are now in the position to state our main theorem.

\begin{theorem}\label{thm:main}
	There exists $\epsilon_0>0$ s.t.\ if $u_0\in B_{l^2_+(\Z,\C^2)}(0,\epsilon_0)$, then there exist $\eta_+\in l^2_+(\Z,\C^2)$, $\rho\geq 0$ and $z\in l^\infty(\Z\cap[0,\infty),\C)$  s.t.
	\begin{align}\label{eq:thm:main:1}
\lim_{t\to\infty}\|u(t)-\Phi_+[z(t)]-U_\infty^{2t}\eta_+\|_{l^2}= 0,
	\end{align}
where $u(t)$ is the solution of nonlinear QW \eqref{eq:nlqw2} and moreover, we have
\begin{align}\label{eq:thm;main:2}
	\lim_{t\to\infty}|z(t)|= \rho \lesssim \|u_0\|_{l^2},\ \|\eta_+\|_{l^2}\lesssim \|P_c u_0\|_{l^2}.
\end{align}
\end{theorem}

The equality \eqref{eq:thm:main:1} states that all small solutions asymptotically become a sum of nonlinear bound state $\Phi_+[z]$ with $|z|=\rho$ and linear scattering wave $U_\infty^{2t}\eta_+$ as expected from soliton resolution conjecture.  
Moreover, from \eqref{eq:thm;main:2} we see that the soliton part $\Phi_+[z]$ converges modulo phase, and the norm of the scattering wave is comparable to the continuous spectrum component $P_cu_0$ of the initial data.

By some additional argument, one can show the "orbital stability" of small nonlinear bound states of nonlinear QWs.
Here, the precise definition of orbital stability is given by the following:
\begin{definition}
Let $\Phi\in l^2$ and $\Lambda\in\R$ and suppose $e^{\im \Lambda t}\Phi$ be a solution of nonlinear QWs \eqref{eq:nlqw2}.
We say $e^{\im \Lambda t}\Phi$ (or simply just $\Phi$) is orbitally stable if for all $\epsilon>0$, there exists $\delta>0$ s.t. if $\|u_0-\Phi\|_{l^2}<\delta$, then $\sup_{t>0} \inf_{\theta\in\R}\|u(t)-e^{\im \theta}\Phi\|_{l^2}<\epsilon$.
\end{definition} 

\begin{corollary}\label{cor:orb}
There exists $\epsilon_0$ s.t. if $|z|<\epsilon_0$, then $e^{\im \Lambda_+[z]}\Phi_+[z]$ is orbitally stable.
\end{corollary}

We now explain the technical hypothesis Assumption \ref{ass:nonlinear} and \ref{ass:linear}.
First, if we set $g(s)= s^p$ (putting aside the fact that $g$ is not smooth when $p>0$ is not an integer), the problem becomes more hard when we let $p$ smaller.
This can be understood from the the trivial inequality $\epsilon^{p_1}<\epsilon^{p_2}$ for $p_1>p_2$ and $\epsilon\in (0,1)$, the embedding $l^2\hookrightarrow l^\infty$ and the fact that we are considering small in $l^2$ solutions.
Our assumption corresponds to $p\geq 3$ in the above nonlinearity.
This restriction is needed to close the estimates using Strichartz estimate (Proposition \ref{eq:stzkato}), which is a consequence of the dispersive estimate $\|U^{t}P_c u_0\|_{l^\infty}\lesssim \<t\>^{-1/3}\|u_0\|_{l^\infty}$ (Proposition \ref{Prop:disp}).
Same constraint is given in the study of asymptotic stability for discrete nonlinear Schr\"odinger equations \cite{CT09SIMA,KPS09SIAM}, which have the same dispersive estimate $\|e^{\im t H_d}P_c u_0\|_{l^\infty}\lesssim \<t\>^{-1/3}\|u_0\|_{l^1}$, where $H_d=-\Delta_d+V$ is the discrete Schr\"odinger operator.
For the nonlinear Schr\"odinger equations such restriction is relaxed to $p\geq 2$ \cite{Mizumachi08JMKU}, due to the fact that in the continuous case, we have better decay $\|e^{\im t H}P_c u_0\|_{L^\infty}\lesssim |t|^{-1/2}\|u_0\|_{L^1}$, where $H=-\Delta+V$ is the Schr\"odinger operator, see \cite{GS04CMP}.

Since the cases $p=1,2$ appear more naturally in physics, there are several attempts to lower $p$.
However, it seems that with current technology, the only way to lower $p$ below 3 is to strengthen the assumption (such as taking $u_0 \in l^{1,1}$) or weakening the result (such as showing the convergence $u(t)-\Phi[z]\to 0$ in a compact domain).
For example, one of the former type result is \cite{MP12DCDS} proving asymptotic stability for $p>2.75$ for discrete nonlinear Schr\"odinger equations by adopting refined decay estimates of Mielke and Patz \cite{MP10AA}.
On the other hand, one of the latter results is \cite{CM19SIMA} proving the convergence in the exponentially weighted space for all $p>0$ for nonlinear Schr\"odinger equation with delta potential using the virial argument developed in the important paper \cite{KMM17JAMS}.
For QWs, we are not aware of any viral type inequalities, thus the next step may be to prove asymptotic stability for $p=2$ with additional assumption on the initial data.
The case $p=1$ seems to be completely out of reach.
From the decay estimate, one expects that the scattering wave needs to be modified already for $p\leq 4/3$ and so $p=1<4/3$ is not even in the threshold of long range scattering.
Thus, such result will be very interesting even if we are considering the case with no bound states.

The assumption that $U$ is generic (i.e.\ $U$ has no edge resonance) is needed for local decay type estimate.
In this paper, following \cite{CT09SIMA} we have used Kato smoothness for the $l^2$-in-time estimate.
Here, Kato smoothness for unitary operator is prepared and since it may have independent interest (for application to smooth scattering theory for QWs, say), we have given the equivalent conditions in the appendix of this paper, see Theorem \ref{thm.KSU.4.1}.

The assumption that $U$ have precisely $2$ eigenvalues is equivalent to $U^2$ having precisely $1$ eigenvalue on $l_+^2(\Z,\C^2)$.
When, there is no eigenvalue, one can show the scattering by the argument of \cite{MSSSS18DCDS} using the estimates given in \cite{MSSSSdis}.
On the other hand, if $U^2$ have more than two eigenvalues on $l_+^2(\Z,\C^2)$, we expect that the quasi-periodic (in time) solutions which exist in the linear case will disappear due to the nonlinear Fermi Golden Rule \cite{Sigal93CMP,SW99IM}.
For results in this direction, see \cite{CM20DCDS} and reference therein.

We next explain the strategy of the proof of the main results.
The proof of the existence of small nonlinear bound states is simple and we just use standard contraction mapping argument.
However, the explicit exponential decay rate for the eigenfunctions given in Proposition \ref{prop:ede} seems not to have appeared in the literature.
Thus, even though the proof is standard, we have put the proof of Proposition \ref{prop:ede} in the appendix of this paper since there may be some independent interest.

The proof of Theorem \ref{thm:main} consists by the following three steps:
\begin{enumerate}
\item Introduce a modulation coordinate $(z,\eta)\mapsto u=\Phi_+[z]+\xi$,
\item prepare linear estimates,
\item and bound $\|\xi\|_{l^p([0,N]\cap\Z,l^{q,\sigma})}$ for some specific $p,q,\sigma$ by induction in $N$,
\end{enumerate}
as established in \cite{SW90CMP} (of course the induction in step 3 is replaced by continuity argument).

In this paper, step.2 is based on the dispersive estimate and the integral representation of the resolvent given in \cite{MSSSSdis}.
This allows us to prepare all necessary estimates by standard duality argument and Christ-Kiselev lemma.
As noted above, for the Kato smoothness $\|U^t P_c u_0\|_{l^2l^{2,-\sigma}}\lesssim \|u_0\|_{l^2}$, we have provided equivalent condition in Theorem \ref{thm.KSU.4.1} and sufficient condition in corollary \ref{cor:KS} in the appendix of this paper.
After the preparation of step 1 and 2, step 3 is more or less standard, we note that the idea using Kato smoothness type estimates comes from \cite{CT09SIMA}.

We would like to emphasis that the novelty of the proof of Theorem \ref{thm:main} is in step 1.
Before explaining our case, we review the modulation argument in the continuous case following \cite{GNT04IMRN}.
First, let us consider the linear Schr\"odinger equation $\im \partial_t u = H u$ where the Schr\"odinger operator having exactly one negative eigenvalue ($-\lambda$).
In this case, it is obvious that the solution having the eigenfunction $\phi$ as an initial data do not decrease because it will evolve as $e^{\im \lambda t}\phi$.
So, it is natural to decompose $u=z \phi + \xi$, where $z=(u,\phi)$ and $\xi =u-(u,\phi)\phi=:P_c u$, where we have normalized $\phi$.
Since the continuous component $\xi$ satisfies $\im \partial_t \xi =HP_c\xi$, we can use the decay estimates of $HP_c$ to show the decay (or scattering) of $\xi$.
Next, we consider the nonlinear Schr\"odinger equation $\im \partial_t u = Hu +g(|u|^2)u$ having nonlinear bound states $e^{\im \lambda[z]t}\phi[z]$ for small $z$ (we have $\phi[z]= z\phi +o(z)$).
As the linear case, we would like to write $u=\phi[z]+\xi$.
The "continuous part" $\xi$ should be determined by some orthogonality condition expressed by a subspace $\mathcal{H}_c\subset L^2$, which should be similar to the linear case $P_cL^2$.
The linearized equation for $\xi$ is $$\partial_t\xi = L_{\mathrm{cont}}[z]\xi,\quad \text{where}\quad L_{\mathrm{cont}}[z]=-\im \(H + g(|\phi(z)|^2)+2g'(|\phi(z)|^2)\Re(\phi(z)\overline{\cdot})\phi(z)\).$$
Since, $L_{\mathrm{cont}}[z]$ is only $\R$-linear due to the complex conjugate $\bar{\cdot}$ and depends on $z$, the "continuous subspace" $\mathcal{H}_c=\mathcal{H}_c[z]$ should also depend on $z$ and it is only required to be $\R$-linear.
The specific choice of the orthogonality condition given in \cite{GNT04IMRN} is
\begin{align}\label{eq:contspacecont}
\mathcal{H}_c[z]:=\{u\in L^2\ |\ \forall w\in\C,\ \<u,\im D\phi[z]w\>=0\},
\end{align}
where $D\phi[z]w=\left.\frac{d}{d\epsilon}\right|_{\epsilon=0}\phi[z+\im w]$ (notice that $D\phi[z]$ is only $\R$-linear).
The continuous space" $\mathcal{H}_c[z]$  can also viewed as the symplectic orthogonal space of the soliton manifold with the symplectic form $\Omega =\<\cdot,\im \cdot\>$ or the condition to eliminate the first order term of $\xi$ in the expansion of the Hamiltonian (or energy) $E(\phi[z]+\xi)$, see e.g.\ \cite{Maeda17SIMA}.
In any case, the point of this choice is that $\mathcal{H}_c[z]$ is near $P_c L^2$ in the sense that $\mathcal{H}_c[0]=P_cL^2$ and moreover it is compatible with the linear evolution of $\xi$.
That is, if $\xi\in \mathcal{H}_c[z]$, we have $\partial_t \xi = L_{\mathrm{cont}}[z] \xi \in \mathcal{H}_c[z]$.

We now come back to the nonlinear QWs.
The difficulty for QWs is that QWs seems not to be Hamilton equations (actually, the author do not even know what "Hamilton equation" means for time discretized system like QWs).
The effect of the lack of Hamiltonian/Energy immediately appears in the study of orbital stability of bound states.
For Schr\"odinger and discrete Schr\"odinger equations the orbital stability of bound states can be proved by the fact that bound states are trapped by the energy under the mass constraint \cite{RW88PD,FO03DIE}.
However, for QWs we cannot use such argument because we are not aware of conservation quantity corresponding to the energy, and in this paper we have proved orbital stability as a consequence of asymptotic stability which give a sieve constraint in the nonlinearity as we have discussed above.
Now, going back to the problem how to choose the "continuous space", first linearizing the dynamics of nonlinear QWs, we have
 $\xi(t+1)=L[z(t)]\xi(t)$, where $L[z]=D\mathcal{U}(\Phi_+[z])$.
 As the nonlinear Schr\"odinger equations, we want to have if $\xi(t)\in \mathcal{H}_c[z(t)]$, then $\xi(t+1)=L[z(t)]\xi(t)\in \mathcal{H}_c[z(t+1)]$.
 The point is that unlike the continuous time case, we cannot fix $t$ and discuss the orthogonality condition at $z(t)$ but we need two points $z(t)$ and $z(t+1)$, where $z(t)$ also evolve in nonlinear manner by the (yet unspecified) orthogonality condition applied to nonlinear QWs.
 The first idea to overcome this difficulty is to replace $z(t+1)$ by $e^{\im \Lambda_+[z(t)]}z(t)$ and treat the remainder $Z(t)=e^{\im \Lambda_+[z(t)]}-z(t+1)$ as an error.
 Indeed, in the end we will be able to show $Z\in l^1$ (see Corollary \ref{cor:main}).  
Once we have replaced $z(t+1)$ by $e^{\im \Lambda_+[z(t)]}z(t)$, what we want becomes $\xi\in \mathcal{H}_c[z]$ implies $L[z]\xi \in \mathcal{H}_c[e^{\im \Lambda_+[z]}z]$.
Although it is not clear the choice \eqref{eq:contspacecont} works, this space have the property $\mathcal{H}_c[e^{\im \theta}z]=e^{\im \theta}\mathcal{H}_c[z]$, so the problem becomes to show $\xi\in \mathcal{H}_c[z]\Rightarrow e^{-\im \Lambda_+[z]}L[z]\xi\in \mathcal{H}_c[z]$.
This property, which is the second difficulty, can be proved by using the "orthogonality" of $L[z]$ with respect to the symplectic form $\<\cdot,\im \cdot\>$.
That is, we prove $\<L[z]\xi_1,\im \xi_2\>=\<\xi_1,\im L[z]^{-1}\xi_2\>$.
This can be thought as an analog relation to the continuous case, where $L_{\mathrm{cont}}[z]$ is symmetric with respect to the symplectic form, that is $\<L_{\mathrm{cont}}[z]\xi_1,\im \xi_2\>=\<\xi_1,\im L_{\mathrm{cont}}[z] \xi_2\>$.
For this property we will rely on the special structure \eqref{eq:N} of the nonlinearity.
It will be interesting to specify the class of gauge invariant nonlinearity which have the above property.
Therefore, although the choice of the continuous space did not have a good reasoning, our choice works and we can complete the modulation argument.

The paper is organized as follows.
In section \ref{sec:nbs} we prove Proposition \ref{prop:nlbs}.
In section \ref{sec:mod}, we give the modulation coordinate.
In section \ref{sec:lin}, we prepare the linear estimates to close the argument.
In section \ref{sec:prmain} we complete the proof of Theorem \ref{thm:main} and Corollary \ref{cor:orb}.
In section \ref{sec:A1} we prove Proposition \ref{prop:ede} and in section \ref{sec:A2}, we provide the necessary and sufficient condition for Kato smoothness.

\section{Nonlinear bound states}\label{sec:nbs}

In this section, we prove Proposition \ref{prop:nlbs}.

First, we claim that as Schr\"odinger (resp.\ discrete Schr\"odinger operator) with $L^1$ ($l^1$) potential, the eigenfunctions corresponding to the discrete spectrum decays exponentially.

\begin{proposition}\label{prop:ede}
	Let $\lambda\in\R$ satisfying $ e^{\im \lambda}\in  \sigma_{\mathrm{d}}(U)$ and let $\phi \in l^2(\Z,\C^2)$ be an eigenfunction of $U$ associated to $e^{\im \lambda}$.
	Then, we have
	\begin{align*}
		\|\phi(x)\|_{\C^2}\sim e^{-\xi(\lambda)|x|},
	\end{align*}
where $\xi(\lambda)>0$ is defined by $\sqrt{1-|\alpha_\infty|^2}\cosh \xi = \cos \lambda$.
\end{proposition}

\begin{remark}\label{rem:disc}
	If $e^{\im \lambda} \in \sigma(U)\setminus \sigma_{\mathrm{ess}} (U)$, then $\cos \lambda >\sqrt{1-|\alpha_\infty|^2}$.
\end{remark}

\begin{remark}
	Proposition only requires $\|C-C_\infty\|_{l^1(\Z,\mathcal{L}(\C^2))}<\infty$ instead of the $l^{1,1}$ boundedness given in Assumption \ref{ass:1}.
	This will be clear from the proof of Proposition \ref{prop:ede}, which will be given in the appendix. 
\end{remark}

Before going into the proof of Proposition \ref{prop:nlbs}, we prepare an elementary lemma.

\begin{lemma}\label{lem:Resbound}
Let $s\geq 0$.
	Under the assumption \eqref{ass:1}, for $\lambda\in \R$ satisfying $e^{\im \lambda}\in \sigma_{\mathrm{d}}(U)$, we have
	\begin{align*}
		\|(U-e^{\im \lambda})^{-1}\tilde{P}_c\|_{\mathcal{L}(l^{2,s}(\Z,\C^2))}\lesssim_s 1,
	\end{align*}
where $\tilde{P}_c:=1-(\cdot,\phi)\phi$ and $\phi$ is the normalized eigenvector of $U$ associated to $e^{\im \lambda}$.
\end{lemma}

\begin{proof}
We argue by induction.
The case $s=0$ is obvious.
If we have the case $s\in\N$, we have the result for all $s>0$ by interpolation.
Thus, we only consider $s\in\N$. 
Let $u\in \tilde{P}_cl^{2,s}$, $f\in \tilde{P}_c l^2$ such that
\begin{align}\label{NLQW8.2}
f=(U-e^{\im \lambda})^{-1}u\ \Leftrightarrow \  (U-e^{\im \lambda})f=u.
\end{align}
We assume that we have the conclusion up to $s-1$, i.e.
\begin{align*}
\|f\|_{l^{2,s-1}}\lesssim_{s-1} \|u\|_{l^{2,s-1}}.
\end{align*}
We set $\<x\>_\epsilon:=\<x\>(1+\epsilon\<x\>)^{-1}$.
Then we have
\begin{align*}
\sup_{x\in \Z}|\<x+1\>_\epsilon^s-\<x\>_\epsilon^s|\lesssim_k \<x\>^{s-1}
\end{align*}
and thus
\begin{align*}
\|[S,\<x\>_\epsilon^k]u(x)\|_{\C^2}\lesssim_k \|\<x\>^{k-1} Su(x)\|_{\C^2}.
\end{align*}
Thus, multiplying $\<x\>_\epsilon$ to both sides of \eqref{NLQW8.2}, we have
\begin{align*}
(U_\infty-e^{\im \lambda})\<x\>_\epsilon^k f= (U-U_\infty)\<x\>_\epsilon^k f+[S,\<x\>_\epsilon^k]Cf+\<x\>_\epsilon^k u,
\end{align*}
where $U_\infty=SC_\infty$,
and thus
\begin{align*}
\|\<x\>_\epsilon^k f\|_{l^2}\lesssim \|\<x\>_\epsilon^{k-1}f\|_{l^2}+\|\<x\>_\epsilon^k u\|_{l^2}.
\end{align*}
Therefore, taking $\epsilon\to 0$, we have the conclusion.
\end{proof}

\begin{proof}[Proof of Proposition \ref{prop:nlbs}]
	We set 
	\begin{align}\label{NLQW1:1}
		\Phi[z]=z\(\phi+|z|^2\psi[|z|^2]\)
	\end{align}
	and
	\begin{align}\label{NLQW1:2}
		\Lambda[z]=\lambda+|z|^2\mu[|z|^2]
	\end{align}
	and look for the solution of \eqref{nlbs:1}.
	Notice that under the above ansatz \eqref{NLQW1:1} and \eqref{NLQW1:2}, \eqref{nlbs:2} and \eqref{nlbs:3} are trivial.
	
	Substituting \eqref{NLQW1:1}, \eqref{NLQW1:2} in \eqref{nlbs:1}, we have
	\begin{align}\label{NLQW1:4}
		\(U-e^{\im \lambda}\)\psi=-U_0\(\frac{N(|z|\(\phi+|z|^2\psi\))-|z|\(\phi+|z|^2\psi\)}{|z|^3}\)+e^{\im \lambda}\frac{e^{\im |z|^2\mu}-1}{|z|^2}\(\phi_0+|z|^2\psi\).
	\end{align}
	Now, taking the inner-product between $\phi$ and assuming $\(\phi,\psi\)=0$, we have
	\begin{align*}
		\(U_0\(\frac{N(|z|\(\phi+|z|^2\psi\))-|z|\(\phi+|z|^2\psi\)}{|z|^3}\),\phi\)=e^{\im \lambda}\frac{e^{\im |z|^2\mu}-1}{|z|^2}.
	\end{align*}
Recall \eqref{eq:N} and setting
	\begin{align*}
		\mathcal{N}(r,\varphi):=&-(U_0-e^{\im \lambda})^{-1}\tilde{P}_cU_0 r^{-1}\(e^{\im g(r\<\phi+r\varphi,\gamma(\phi+r\varphi)\>_{\C^2})\gamma}-1\)(\phi+r\varphi),
	\end{align*}
	where $\tilde{P}_c$ is the orthogonal projection to the orthogonal complement of $\phi$, we can write \eqref{NLQW1:4} as
	\begin{align*}
		\psi=\mathcal{N}(|z|^2,\psi).
	\end{align*}
Further, we have $\mathcal{N}\in C^\infty(\R\times l^{2,s},l^{2,s})$.
	By implicit function theorem, it suffices to show that there exist $\delta_s>0$ and $C_s>0$ s.t.\ for $r\leq \delta_s^2$
	\begin{align*}
		\mathcal{N}(r,\cdot):\overline{B_{\tilde{P}_cl^{2,s}(\Z^d,\C^N)}(0,C_s)}\to \overline{B_{\tilde{P}_cl^{2,s}(\Z^d,\C^N)}(0,C_s)}
	\end{align*}
	and
	\begin{align}\label{eq:Ndiff}
		\| \mathcal{N}(r,\varphi_1)-\mathcal{N}(r,\varphi_2)\|_{l^{2,s}}\leq \frac{1}{2}\|\varphi_1-\varphi_2\|_{l^{2,s}}
	\end{align}
	
	We set $C_s:=2\max_{|r|\leq 1}\|\mathcal{N}(r,0)\|_{l^{2,s}}$.
	Next, take $\varphi_1,\varphi_2\in \overline{B_{\tilde{P}_cl^{2,s}(\Z^d,\C^N)}(0,C_s)}$ and set
	\begin{align*}
		f(\tau)=r^{-1}\(e^{\im g(r\<\phi_\tau,\gamma\phi_\tau\>_{\C^2})\gamma}-1\)\phi_\tau,
	\end{align*}
where $\phi_\tau = \phi+r(\varphi_2+\tau(\varphi_1-\varphi_2))$.
Then, from Lemma \ref{lem:Resbound}, we have
\begin{align*}
	\| \mathcal{N}(r,\varphi_1)-\mathcal{N}(r,\varphi_2)\|_{l^{2,s}}\lesssim_s \|f(1)-f(0)\|_{l^{2,s}}.
\end{align*}
Computing the derivative of $f$, we have
\begin{align*}
	f'(\tau)=&r^{-1}\(e^{\im g(r\<\phi_\tau,\gamma\phi_\tau\>_{\C^2})\gamma}-1\)r(\varphi_1-\varphi_2)\\&+\im re^{\im g(r\<\phi_\tau,\gamma\phi_\tau\>_{\C^2}) \gamma}g'(r\<\phi_\tau,\gamma\phi_\tau\>_{\C^2}) \<\phi_\tau, \varphi_1-\varphi_2\>_{\C^2}\gamma \phi_\tau.
\end{align*}
Thus, by $f(1)-f(0)=\int_0^1f'(\tau)\,d\tau$, we have
	\begin{align*}
		\|f(1)-f(0)\|_{l^{2,s}}\lesssim_s r \|\varphi_1-\varphi_2\|_{l^{2,s}}.
	\end{align*}
Therefore, taking $\delta_s$ sufficiently small, we have \eqref{eq:Ndiff}.

Finally, we prove $\Lambda[z]\in \R$.
By the above, we have \eqref{nlbs:1}.
Then, taking the $l^2$ norm of both sides, we see $|e^{\im \Lambda[z]}|=1$.
Thus, we have the conclusion.
\end{proof}

Notice that by the fact that $\Phi[z]\in  l^{2}$ and $g$ is smooth with $g(0)=0$, we can show $\|C[z]-C_\infty\|_{l^1}<\infty$ where $C[z]:=Ce^{\im g(\<\Phi[z],\gamma\Phi[z]\>_{\C^2})\gamma}$.
Because $C[z]\Phi[z]=e^{\im \Lambda[z]}\Phi[z]$, we can apply Proposition \ref{prop:ede} and thus we have the following sharp exponential decay estimate of the nonlinear bound state:
\begin{align*}
\|\Phi[z](x)\|_{\C^2}\sim e^{-\xi(\Lambda[z])|x|}.
\end{align*}
However, we will not use this estimate.

\section{Modulation equation}\label{sec:mod}

As discussed in the introduction, we set the continuous space $\mathcal{H}_c[z]$ by
\begin{align*}
	\mathcal{H}_c[z]:=\{u\in l^2(\Z,\C^2)\ |\ \forall w\in \C,\ \<u,\im D\Phi_+[z]w\>=0\}.
\end{align*}

We first show that all small in $l^2$ functions can be decompose as $u=\Phi_+[z]+\xi$ with $\xi\in \mathcal{H}_c[z]$ by standard implicit function theorem argument.

\begin{lemma}\label{lem:mod}
	There exists $\epsilon>0$ s.t.\ for $u\in B_{l^2_+(\Z,\C^2)}(0,\epsilon)$, there exists $z\in \C$ s.t.
	\begin{align*}
		u-\Phi_+[z]\in \mathcal{H}_c[z].
	\end{align*}
\end{lemma}

\begin{proof}
	Set
	\begin{align*}
	F(z,u):=\begin{pmatrix}
	\<u-\Phi_+[z],\im D\Phi_+[z]1\>\\
	\<u-\Phi_+[z],\im D\Phi_+[z]\im \>
	\end{pmatrix}.
	\end{align*}
	By the definition of $F$, the problem is reduced to find $z$ s.t.\ $F(z,u)=0$ for given small $u$.
	However, this follows easily from implicit function theorem.
	Indeed, we have
	\begin{align*}
	\left.\frac{\partial F(z,u)}{\partial(z_R,z_I)}\right|_{(z,u)=(0,0)}=\begin{pmatrix}
	\<-D\Phi_+[0]1,\im D\Phi_+[0]1\>\ \<-D\Phi_+[0]\im,\im D\Phi_+[0]1\>\\
	\<-D\Phi_+[0]1,\im D\Phi_+[0]\im\>\ \<-D\Phi_+[0]\im,\im D\Phi_+[0]\im\>
	\end{pmatrix}
	=
	\begin{pmatrix}
	0 & -1\\
	1 & 0
	\end{pmatrix}.
	\end{align*}
	Thus, we have the conclusion.
\end{proof}

By lemma \ref{lem:mod}, given a solution $u(t)$ of nonlinear QW \eqref{eq:nlqw2}, we can write
\begin{align}\label{eq:mod1}
	u(t)=\Phi_+[z(t)]+\xi(t),\ \xi(t)\in \mathcal{H}_c[z(t)].
\end{align}
We set
\begin{align*}
	L[z]:=D\mathcal{U}(\Phi_+[z]).
\end{align*}
Since $L[0]=U^2$, there exists $\delta>0$ s.t. if $|z|<\delta$, then $L[z]$ is invertible.
In the following, if necessary we replace $\delta_0$ by the above $\delta$ and assume $L[z]$ is always invertible.
Notice that $L[z]$ is not $\C$-linear but only $\R$-linear.
$L[z]$ is "orthogonal" w.r.t.\ the symplectic form $\<\cdot,\im\cdot\>$ in the following sense:
\begin{lemma}\label{lem:symporth}
	We have
	\begin{align*}
		\<L[z]u,\im v\>=\<u,\im L[z]^{-1}v\>.
	\end{align*}
\end{lemma}

\begin{proof}
	We first show that for $u,v\in \C^2$ and $w\in B_{\C^2}(0,\delta)$ with sufficiently small $\delta$,
	\begin{align}\label{eq:pr:symporth1}
		\<DN(w)u,\im v\>_{\C^2}=\<u,\im DN(w)^{-1}v\>_{\C^2}.
	\end{align}
We have
\begin{align}\label{eq:pr:symporth2}
	DN(w)u=e^{\im g(\<w,\gamma w\>_{\C^2})\gamma}\(u+2g'(\<w,\gamma w\>_{\C^2})\<u,\gamma\phi\>_{\C^2}\im \gamma w\).
\end{align}
Now, set
\begin{align*}
	A(w)u:=2g'(\<w,\gamma w\>_{\C^2})\<u,\gamma w\>_{\C^2}\im \gamma w.
\end{align*}
Then, we have
\begin{align*}
	A(w)^2u&=2g'(\<w,\gamma w\>_{\C^2})\<A(w)u,\gamma w\>_{\C^2}\im \gamma w
	=4\(g'(\< w,\gamma w\>_{\C^2})\)^2\<\im \gamma w,\gamma w\>_{\C^2}\im \gamma w=0.
\end{align*}
Thus, we have
\begin{align}\label{eq:pr:symporth3}
	DN(w)^{-1}=(1-A(w))e^{-\im g(\<w,\gamma w\>_{\C^2})\gamma }.
\end{align}
Substituting \eqref{eq:pr:symporth2} and \eqref{eq:pr:symporth3} into \eqref{eq:pr:symporth1} and comparing the both side we can verify \eqref{eq:pr:symporth1}.

Next, since $L[z]=UDN\(UN\(\Phi_+[z]\)\)UDN(\Phi_+[z])$ and $U$ satisfies $\<Uu,\im v\>=\<u,\im U^{-1}v\>$ for $u,v\in l^2(\Z,\C^2)$, we have
\begin{align*}
	\<L[z]u,\im v\>=\<u,\im DN(\Phi_+[z])^{-1}U^{-1}DN\(UN\(\Phi_+[z]\)\)^{-1}U^{-1}v\>.
\end{align*}
Since $L[z]^{-1}=DN(\Phi_+[z])^{-1}U^{-1}DN\(UN\(\Phi_+[z]\)\)^{-1}U^{-1}$, we have the conclusion.
\end{proof}

Before writing down the equations of $z(t)$ and $\xi(t)$, we prepare a lemma
\begin{lemma}\label{lem:orth}
	Let $z\in B_{\C}(0,\delta)$ and $\xi \in \mathcal{H}_c[z]$.
	Then, $e^{-\im \Lambda_+[z]}L[z]\xi \in \mathcal{H}_c[z]$.
\end{lemma}

\begin{proof}
Take arbitrary $w\in \C$.
	Differentiating \eqref{eq:Phiplus} w.r.t.\ $z$ in $w$ direction, we have
	\begin{align}\label{eq:pr:orth1}
		L[z]D\Phi_+[z]w=e^{\im \Lambda_+[z]}\(D\Phi[z]w + (D\Lambda_+[z]w)\im \Phi_+[z]\).
	\end{align}
Next, by differentiating the first equation of \eqref{nlbs:2}, we have
\begin{align}\label{eq:pr:orth2}
	D\Phi_+[z]\im z = \im \Phi_+[z].
\end{align}
Substituting \eqref{eq:pr:orth2} into \eqref{eq:pr:orth1}, we have
\begin{align*}
	L[z]D\Phi_+[z]w=e^{\im \Lambda_+[z]}D\Phi[z]\(w + \im(D\Lambda_+[z]w) z\).
\end{align*}
Now, since the $\R$-linear map $w\mapsto w + \im(D\Lambda_+[z]w) z$ is invertible for small $z$, we have
\begin{align*}
	D\Phi_+[z](1+\im (D\Lambda_+[z]\cdot)z)^{-1}w=L[z]^{-1}e^{\im \Lambda_+[z]}D\Phi[z]w.
\end{align*}
Thus, by Lemma \ref{lem:symporth} we have
\begin{align*}
	\<e^{-\im \Lambda_+[z]}L[z]\xi,\im D\Phi_+[z]w\>=\<\xi,\im L[z]^{-1}e^{\im \Lambda_+[z]}D\Phi_+[z]w\>=\<\xi,\im D\Phi_+[z](1+\im (D\Lambda_+[z]\cdot)z)^{-1}w\>=0.
\end{align*}
Therefore, we have the conclusion.
\end{proof}

Substituting \eqref{eq:mod1} in \eqref{eq:nlqw2}, we have
\begin{align}\label{eq:xi}
	\xi(t+1)=L[z(t)]\xi(t)+F_1[z(t),z(t+1)]+G[z(t),\xi(t)],
\end{align}
where
\begin{align*}
	F_1[z_1,z_2]&:=e^{\im \Lambda_+[z_1]}\Phi_+[z_1]-\Phi_+[z_2],\\
	G[z,\xi]&:=\mathcal{U}(\Phi_+[z]+\xi)-\mathcal{U}(\Phi_+[z])-L[z]\xi.
\end{align*}
Let $w\in \C$.
Substituting \eqref{eq:xi} in $\<\xi(t+1),\im D\Phi_+[z(t+1)]w\>=0$, we have
\begin{align*}
	0=\<L[z(t)]\xi(t)+F_1[z(t),z(t+1)]+G[z(t),\xi(t)],\im D\Phi_+[z(t+1)]w\>.
\end{align*}
By $D\Phi_+[e^{\im \theta} z]w=e^{\im \theta}D\Phi_+[z]e^{-\im \theta}w$, which follows from \eqref{nlbs:2}, and lemma \ref{lem:orth} we have
\begin{align*}
	\<L[z(t)]\xi(t),D\Phi_+[z(t+1)]w\>=\<L[z(t)]\xi(t),\(D\Phi_+[z(t+1)]-D\Phi_+[e^{\im \Lambda[z(t)]}z(t)]\)w\>.
\end{align*}
Next, setting
\begin{align}\label{eq:defZ}
Z(t):=e^{\im \Lambda_+[z(t)]}z(t)-z(t+1),
\end{align}
and
 $w=\im Z(t)$, we have
\begin{align*}
\<F_1[z(t),z(t+1)],\im D\Phi[z(t+1)]w\>=-|Z(t)|^2+g(t),
\end{align*}
where
\begin{align}
	g(t)= \<F_1[z(t),z(t+1)] -Z(t)\phi,\im D\Phi[z(t+1)]\im Z(t)\> +\<Z(t)\phi,\im \(D\Phi_+[z_2]-\phi\)\im Z(t)\>.\label{eq:F2}
\end{align}
Therefore, we obtain
\begin{align}
|Z(t)|^2=&\<G[z(t),\xi(t)],\im D\Phi_+[z(t+1)]\im Z(t)\> +g(t)\label{eq:z2}\\&	+\<L[z(t)]\xi(t),\(D\Phi_+[z(t+1)]-D\Phi_+[e^{\im \Lambda[z(t)]}z(t)]\)\im Z(t)\>.\nonumber
\end{align}

We introduce the inverse of $P_{c}$ on $\mathcal{H}_c[z]$.

\begin{lemma}\label{lem:modcoor}
Let $s>0$.
There exists $\delta_s>0$ s.t.\ 
there exists $a_R,a_I \in C^\infty(B_{\C}(0,\delta_s),l^{2,s}(\Z,\C^2))$ s.t. 
\begin{align}\label{GNTop.est.1}
|a_R(z)|+|a_I(z)|\lesssim |z|^2.
\end{align}
Moreover, setting $R[z]$ by
	\begin{align*}
	R[z]\eta=\eta+\<\eta,a_R[z]\> \phi_++\<\eta,a_I[z]\>\im \phi_+,
	\end{align*}
	we have $R[z]:P_c l^2_+(\Z,\C^2)\to \mathcal{H}_c[z]$, $P_c R[z]=\left.I_d\right|_{P_c l_+^2(\Z,\C^2)}$ and $R[z]P_c =\left.\mathrm{Id}\right|_{\mathcal{H}_c[z]}$.
\end{lemma}

\begin{proof}
We look for $a_R[z]$, $a_I[z]$ satisfying
\begin{equation}\label{GNTop.pr.1}
\begin{aligned}
\<\eta+\<\eta,a_R[z]\>\phi_++\<\eta,a_I[z]\>\im\phi_+,\im D\Phi_+[z]1\>&=0,\\
\<\eta+\<\eta,a_R[z]\>\phi_++\<\eta,a_I[z]\>\im\phi_+,\im D\Phi_+[z]\im\>&=0,
\end{aligned}
\end{equation}
%
Since
\begin{align*}
\left.\begin{pmatrix}
\<\phi_+,\im D\Phi_+[z]1\> & \<\im\phi_+,\im D\Phi_+[z]1\>\\
\<\phi_+,\im D\Phi_+[z]\im\> & \<\im\phi_+,\im D\Phi_+[z]\im\>
\end{pmatrix}\right|_{z=0}=\begin{pmatrix}
0 & 1\\ -1 & 0
\end{pmatrix},
\end{align*}
is invertible, the matrix above is also invertible for small $z$.
Solving the above w.r.t.\ $\<\eta,a_R[z]\>$ and $\<\eta,a_I[z]\>$, we have
\begin{align*}
\begin{pmatrix}
\<\eta,a_R[z]\>\\
\<\eta,a_I[z]\>
\end{pmatrix}
=-\begin{pmatrix}
\<\phi_+,\im D\Phi_+[z]1\> & \<\im\phi_+,\im D\Phi_+[z]1\>\\
\<\phi_+,\im D\Phi_+[z]\im\> & \<\im\phi_+,\im D\Phi_+[z]\im\>
\end{pmatrix}^{-1}
\begin{pmatrix}
\<\eta,\im D\Phi_+[z]1\>\\
\<\eta,\im D\Phi_+[z]\im \>
\end{pmatrix}.
\end{align*}
Thus, taking
\begin{align}\label{GNTop.pr.2}
\begin{pmatrix}
a_R[z]\\
a_I[z]
\end{pmatrix}
=-\begin{pmatrix}
\<\phi_+,\im D\Phi_+[z]1\> & \<\im\phi_+,\im D\Phi_+[z]1\>\\
\<\phi_+,\im D\Phi_+[z]\im\> & \<\im\phi_+,\im D\Phi_+[z]\im\>
\end{pmatrix}^{-1}
\begin{pmatrix}
\im D\Phi_+[z]1\\
\im D\Phi_+[z]\im
\end{pmatrix},
\end{align}
we see that $a_R$, $a_I$ have the desired properties.
We remark that the property $R[z]P_c =\left.I_d\right|_{\mathcal{H}_c[z]}$ follows from the uniqueness of the solution of \eqref{GNTop.pr.1} and the estimate \eqref{GNTop.est.1} follows from \eqref{GNTop.pr.2} and \eqref{nlbs:3}.
\end{proof}

By lemma \ref{lem:modcoor}, setting $\eta(t):=P_c\xi(t)\in P_cl^2_+(\Z,\C^2)$, we have
\begin{align}\label{eq:modcoor}
	u(t)=\Phi_+[z(t)]+R[z(t)]\eta(t).
\end{align}
Applying $P_c$ to \eqref{eq:xi}, we have
\begin{align}\label{eq:eta}
	\eta(t+1)=U^2\eta(t) + P_c F[z(t),z(t+1), \eta(t)],
\end{align}
where $F[z_1,z_2,\eta]=F_1[z_1,z_2]+F_2[z_1,\eta]+F_3[z_1,\eta]$ and
\begin{align}\label{eq:defF2F3}
F_2[z_1,\eta]:=\(L[z_1]-U^2\)R[z_1]\eta,\quad
F_3[z_1,\eta]:= G[z_1,R[z_1]\eta].
\end{align}
By Duhamel's formula, \eqref{eq:eta} can written as 
\begin{align*}
\eta(t)=U^{2t}\eta(0)+\sum_{s=0}^{t-1}U^{2(t-s)}P_c F[z(s),z(s+1),\eta(s)].
\end{align*}

\section{Linear estimates}\label{sec:lin}

In this section, we collect the linear estimate which we need for the bootstrap argument.
Since this section only deals with linear estimate, we will only assume
\begin{itemize}
\item Assumption \ref{ass:1} and
\item $U$ is generic in the sense of Definition 1.4 of \cite{MSSSSdis}. 
\end{itemize}
Thus, we will not assume that $U$ has exactly one eigenvalue ($U$ can have no eigenvalue or many eigenvalues, it is known that under the above assumption $U$ has only finitely many eigenvalues, see Proposition 1.8 of \cite{MSSSSdis}).
We set $P_c(U)$ to be the orthogonal projection to the orthogonal complement of the eigenvectors of $U$.

\begin{proposition}[Dispersive estimate]\label{Prop:disp}
We have
\begin{align*}
\|U^t P_c(U) u_0\|_{l^\infty}\lesssim \<t\>^{-1/3}\|u_0\|_{l^1}.
\end{align*}
\end{proposition}

\begin{proof}
See, \cite{MSSSSdis}.
\end{proof}

We also use the Kato smoothness property (actually the sufficient condition of Kato smoothness \cite{Kato65MA}).
\begin{lemma}\label{lem:suffKato}
Let $s>1$.
Then, 
\begin{align*}
\sup_{\Im \mu\neq 0}\|R(\mu)P_c(U)\|_{\mathcal{L}(l^{2,s},l^{2,-s})}<\infty.,
\end{align*}
where $R(\mu)=(Ue^{-\im \mu}-1)^{-1}$.
\end{lemma}

\begin{proof}
By (5.13) of \cite{MSSSSdis} (we note that there is a typo and the r.h.s.\ is $\max$ instead of $\min$) and Lemma 6.5 of \cite{MSSSSdis}, the kernel of $R(\mu)$ can decomposed as
\begin{align*}
K(x,y)=K_{1}(x,y)+K_2(x,y),
\end{align*}
s.t.
\begin{align*}
\|K_1(x,y)\|_{\mathcal{L}(\C^2)}\lesssim \max(1,x)\max(1,-y)1_{x\leq y},\\
\|K_2(x,y)\|_{\mathcal{L}(\C^2)}\lesssim \max(1,-x)\max(1,y)1_{x\geq y},
\end{align*}
where the implicit constant is independent of $\mu\in \C\setminus\R$.
Since the two are symmetric, we only bound the first.
First, for $y<0$,
\begin{align*}
\sum_{x\in\Z}\<x\>^{-s}\|K_1(x,y)\|_{\mathcal{L}(\C^2)}\<y\>^{-s}\lesssim\sum_{x=-\infty}^{y}\<x\>^{-s}\<y\>^{1-s}\sim \<y\>^{2-s}\lesssim 1.
\end{align*}
Next, for $y\geq 0$,
\begin{align*}
\sum_{x\in\Z}\<x\>^{-s}\|K_1(x,y)\|_{\mathcal{L}(\C^2)}\<y\>^{-s}\lesssim\sum_{x=-\infty}^{-1}\<x\>^{-s}\<y\>^{1-s}+\sum_{x=0}^y\<x\>^{1-s}\<y\>^{-s}
\lesssim 1+ \<y\>^{2-2s}\lesssim 1.
\end{align*}
Thus, we have the desired bound from the Schur test.
\end{proof}

For $I\subset \Z$, we set,
\begin{align*}
\mathrm{Stz}(I)&:=l^6(I,l^\infty(\Z,\C^2))\cap l^\infty(I,l^2(\Z,\C^2)),\\
\mathrm{Stz}^*(I)&:=l^{6/5}(I,l^1(\Z,\C^2))+ l^1(I,l^2(\Z,\C^2)).
\end{align*}
Further, we set $\mathrm{Stz}_T:=\mathrm{Stz}(\Z\cap [0,T])$ and $\mathrm{Stz}:=\mathrm{Stz}(\Z\cap [0,\infty))$.
Similarly, we set $l^p_TX:=l^p(\Z\cap[0,T],X)$ and $l^pX:=l^p(\Z\cap[0,\infty),X)$.

\begin{proposition}[Strichartz estimates and Kato smoothness]\label{prop:homest}
Let $s>1$ and $0\leq t_1<t_2$.
Then, we have
\begin{align}
\|U^t P_c(U) u_0\|_{\mathrm{Stz}_T\cap l^2_Tl^{2,-s}}&\lesssim_s \|u_0\|_{l^2},\label{eq:stzkato}\\
\|\sum_{\tau=t_1}^{t_2-1} U^{-\tau}P_c(U) F(\tau)\|_{l^2}&\lesssim_s \|F\|_{\mathrm{Stz}^*\([t_1,t_2-1]\) + l^2\([t_1,t_2-1],l^{2,s}\)}. \label{eq:dualest}
\end{align}
where the implicit constant is independent of $T$.
\end{proposition}

\begin{proof}
The homogeneous Strichartz estimates $\|U^t P_c(U) u_0\|_{\mathrm{Stz}_T}\lesssim \|u_0\|_{l^2}$ follows from Proposition \ref{Prop:disp} combined with standard argument, see e.g. the proof of Lemma 2.3 of \cite{MSSSS18DCDS}.
The estimate $\|U^t P_c(U)u_0\|_{l^2_T l^{2,-s}}\lesssim_s \|u_0\|_{l^2}$ is an unitary analog of the Kato smoothness of $\<x\>^{-s}P_c(U)$, see Definition \ref{def:KS}.
Thus, from Corollary \ref{cor:KS} and Lemma \ref{lem:suffKato} we have the estimate.
Finally, the estimate \eqref{eq:dualest} is the dual of \eqref{eq:stzkato}.
\end{proof}

\begin{proposition}\label{prop:inhomest}
Let $s>1$.
Then, we have
\begin{align}
\|\sum_{\tau=0}^{t-1} U^{t-\tau}P_c(U)F(\tau)\|_{\mathrm{Stz}_{T+1}}&\lesssim_s \|F\|_{\mathrm{Stz}^*_T+l^2_Tl^{2,s}},\label{eq:inhom1}\\
\|\sum_{\tau=0}^{t-1} U^{t-\tau}P_c(U)F(\tau)\|_{l^2_{T+1}l^{2,-s}}&\lesssim_s \|F\|_{l^2_Tl^{2,s}},\label{eq:inhom2}
\end{align}
where the implicit constant is independent of $T$.
\end{proposition}

\begin{proof}
The estimate \eqref{eq:inhom1} follows from Christ-Kiselev lemma  combined with \eqref{eq:stzkato} and \eqref{eq:dualest}.
For \eqref{eq:inhom2}, restricting the operator appropriately, it suffices to show
\begin{align*}
\|\sum_{\tau=-\infty}^{t-1} e^{-\epsilon(t- \tau)} U^{t-\tau}P_c(U)F(\tau)\|_{l^2\(\Z,l^{2,-s}\)}&\lesssim_s \|F\|_{l^2\(\Z,l^{2,s}\)},
\end{align*}
with the implicit constant independent of $\epsilon$.
By Plancherel and Fubini, we have
\begin{align*}
&\|\sum_{\tau=-\infty}^{t-1} e^{-\epsilon(t- \tau)} U^{t-\tau}P_c(U)F(\tau)\|_{l^2(\Z,l^{2,-s})}
=\| \mathcal{F}\(\sum_{\tau=-\infty}^{\cdot-1}(e^{-\epsilon}U)^{\cdot-\tau}P_c(U) F(\tau)\)(\lambda)\|_{L^2_\lambda(\T,l^{2,-s})}
\\&
=\frac{1}{2\pi}\| \sum_{\tau\in\Z}\sum_{t=\tau+1}^{\infty}e^{-\im \lambda(t-\tau) } (e^{-\epsilon}U)^{t-\tau}P_c(U) e^{\im \lambda \tau}F(\tau)\|_{L^2(\T,l^{2,-s})} \\&
=\frac{1}{2\pi}e^{-\epsilon}\|R(\lambda-\im \epsilon)P_c(U)\sum_{\tau\in\Z}e^{\im \lambda \tau}F(\tau)\|_{L^2(\T,l^{2,-s})}\\&
\lesssim \sup_{\Im \mu\neq 0}\|R(\mu)P_c(U)\|_{\mathcal{L}(l^{2,s},l^{2,-s})}\|\mathcal{F}\(F(-\cdot)\)(\lambda)\|_{L^2(\T,l^{2,s})}\lesssim \|F\|_{l^2(\Z,l^{2,s})}.
\end{align*}
Here, the Fourier transform is given by $\mathcal{F}u(\lambda)=(2\pi)^{-1}\sum_{t\in\Z}e^{-\im \lambda t}u(t)$.
\end{proof}
We remark that even though we have prepared the estimate for $U$, it is easy to modify to obtain the same estimate for $U^2$.

\section{Proof of the main theorem}\label{sec:prmain}
In this section, we prove Theorem \ref{thm:main}.
First, from Proposition \ref{prop:nlbs}, Lemma \ref{lem:modcoor} and \eqref{eq:modcoor}, we have
\begin{align*}
\|u(t)\|_{l^2}\sim |z(t)|+\|\eta(t)\|_{l^2}.
\end{align*}
Let $\epsilon_0>0$ to be a sufficiently small constant which we will determine later and set
\begin{align*}
\|u_0\|_{l^2}=\epsilon <\epsilon_0.
\end{align*}
Then, since $\|u(t)\|_{l^2}=\|u_0\|_{l^2}$, we have
\begin{align}\label{eq:aprioribound}
\|z\|_{l^\infty(\Z)}+\|\eta\|_{l^\infty(\Z,l^2(\Z,\C^2))}\sim \epsilon.
\end{align}

We will prove the following bootstrap result:
\begin{proposition}\label{prop:boost}
There exists $C_0>0$ s.t. for $C_1\geq C_0$, there exists $\epsilon_1>0$ s.t.\ 
if $\|u_0\|_{l^2}=\epsilon<\epsilon_1$ and
\begin{align}\label{eq:prop:boost}
\|\eta\|_{\mathrm{Stz}_T\cap l^2_Tl^{2,-2}}\leq C_1\|\eta(0)\|_{l^2},
\end{align}
then \eqref{eq:prop:boost} holds for $T$ replaced by $T+1$.
\end{proposition}

To prove Proposition \ref{prop:boost}, we prepare several lemmas.
\begin{lemma}\label{lem:boot1}
There exists $\epsilon_0>0$ such that, under the assumption of Proposition \ref{prop:boost}, we have
\begin{align*}
\|\sum_{\tau=0}^{t-1}U^{2(t-\tau)}P_c F_2[z(\tau),\eta(\tau)]\|_{\mathrm{Stz}_{T+1}\cap l^2_{T+1}l^{2,2}}\lesssim \epsilon^2\|\eta(0)\|_{l^2}.
\end{align*}
\end{lemma}

\begin{proof}
First, it is easy to show 
\begin{align*}
\|U\|_{\mathcal{L}(l^{2,s})}+\|DN\(UN(\Phi_+[z])\)\|_{\mathcal{L}(l^{2,s})}+\|DN\(\Phi_+[z]\)\|_{\mathcal{L}(l^{2,s})}+\|R[z]\|_{\mathcal{L}(l^{2,s})}\lesssim_s 1,
\end{align*}
for $s>0$.
Further, by Proposition \ref{prop:nlbs} and \eqref{eq:pr:symporth2}, we have
\begin{align*}
\|DN(UN(\Phi_+[z]))-1\|_{\mathcal{L}(l^{2,-s},l^{2,s})}+\|DN(\Phi_+[z])-1\|_{\mathcal{L}(l^{2,-s},l^{2,s})}\lesssim_s |z|^2,
\end{align*}
where $z$ is taken sufficiently small depending on $s$.
Thus, we have
\begin{align}
\|F_2[z,\eta]\|_{l^{2,s}}\lesssim_s &\|U \(DN(UN(\Phi_+[z]))-1\)UDN(\Phi_+[z])R[z]\eta\|_{l^{2,s}}\nonumber\\&+\|U^2\(DN(\Phi_+[z])-1\)R[z]\eta\|_{l^{2,s}}
\lesssim_s |z|^2 \|\eta\|_{l^{2,-s}}.\label{eq:F2bound}
\end{align}
Therefore, from Proposition \ref{prop:inhomest}, we have the conclusion. 
\end{proof}

\begin{lemma}\label{lem:boot3}
	Under the assumption of Proposition \ref{prop:boost}, we have
\begin{align*}
\|\sum_{\tau=0}^{t-1}U^{2(t-\tau)}P_c F_3[z(\tau),\eta(\tau)]\|_{\mathrm{Stz}_{T+1}\cap l^2_{T+1}l^{2,2}}\lesssim \(C_1^7\epsilon^6 + C_1^2\epsilon^2\) \|\eta(0)\|_{l^2}.
\end{align*}
\end{lemma}

\begin{proof}
Set
\begin{align}\label{eq:defF31F32}
F_{31}[\eta]:=F_3[0,\eta]\text{ and }F_{32}[z,\eta]:=F_3[z,\eta]-F_{3}[0,\eta].
\end{align}
We first estimate the contribution of $F_{31}$.
By Assumption \ref{ass:nonlinear}, for $w\in B_{\C^2}(0,1)$, we have
\begin{align*}
\|\(e^{\im g(\<w,\gamma w\>_{\C^2})\gamma}-1\)w\|_{\C^2}\lesssim \|w\|_{\C^2}^7,
\end{align*}
which implies
\begin{align*}
\|F_{31}[\eta(t)]\|_{l^2}&\leq \|U\(N(UN(\eta(t)))-UN(\eta(t))\)\|_{l^2}+\|U^2(N(\eta(t))-\eta(t))\|_{l^2}\\&
\lesssim \||UN(\eta(t))|^7\|_{l^2}+\||\eta(t)|^7\|_{l^2}=\|UN(\eta(t))\|_{l^{14}}^7+\|\eta(t)\|_{l^{14}}^7\lesssim \|\eta(t)\|_{l^{14}}^7.
\end{align*}
Thus, by $\mathrm{Stz}\hookrightarrow l^7 l^{14}$, we have
\begin{align}\label{eq:F31est}
\|\sum_{\tau=0}^{t-1}U^{2(t-\tau)}P_c F_{31}[\eta(\tau)]\|_{\mathrm{Stz}_{T+1}}\lesssim \| F_{31}[\eta(\cdot)]\|_{l^1_Tl^2}\lesssim \|\eta\|_{l^7l^{14}}^7\lesssim \|\eta\|_{\mathrm{Stz}}^7.
\end{align}
On the other hand, from Minkowski inequality, \eqref{eq:stzkato} and \eqref{eq:F31est}, we have
\begin{align*}
&\|\sum_{\tau=0}^{t-1}U^{2(t-\tau)}P_c F_{31}[\eta(\tau)]\|_{l^2_{T+1}l^{2,-s}}\lesssim \|\sum_{\tau=0}^T \|U^{2(t-\tau)}P_cF_{31}[\eta(\tau)]\|_{\C^2}\|_{l^2_{T+1}l^{2,-s}}\nonumber\\&
\lesssim \sum_{\tau=0}^T \|U^{2t}P_cF_{31}[\eta(\tau)]\|_{l^2l^{2,-s}} \lesssim 
\sum_{\tau=0}^T \|F_{31}[\eta(\tau)]\|_{l^2}\lesssim \|\eta\|_{\mathrm{Stz}}^7.
\end{align*}

We next investigate the contribution of $F_{32}$.
Set
\begin{align*}
f(t,\tau)=\mathcal{U}(u_{t,\tau}),\ u(t,\tau):=\Phi_+[tz]+\tau R[tz]\eta.
\end{align*}
Then, we have
\begin{align*}
F_{32}[z,\eta]=f(1,1)-f(1,0)-\partial_\tau f(1,0)-f(0,1)=\int_0^1\int_0^1(1-\tau)\partial_t\partial_\tau^2f(t,\tau)\,dtd\tau.
\end{align*}
Further, recalling $\mathcal{U}=UN(UN(\cdot))$, and setting $v_{t,\tau}=UN(u_{t,\tau})$, $w_{t,\tau}=D\Phi_+[tz]z+\tau(DR[tz]z)\eta$, we have
\begin{align*}
&\partial_t\partial_\tau^2f(t,\tau)
=UD^3N(v_{t,\tau})\(UDN(u_{t,\tau})w_{t,\tau},UDN(u_{t,\tau})R[tz]\eta,UDN(u_{t,\tau})R[tz]\eta\)\\&
+2UD^2N(v_{t,\tau})\(UDN(u_{t,\tau})R[tz]\eta,UD^2N(u_{t,\tau})\(R[tz]\eta,w_{t,\tau}\)+UDN(u_{t,\tau})\(DR[tz]z\)\eta\)\\&
+UD^2N(v_{t,\tau})\(UDN(u_{t,\tau})w_{t,\tau},UD^2N(u_{t,\tau})\(R[tz]\eta,R[tz]\eta\)\)\\&
+UDN(v_{t,\tau})UD^3N(u_{t,\tau})(w_{t,\tau},R[tz]\eta,R[tz]\eta)\\&
+2UDN(v_{t,\tau})UD^2N(u_{t,\tau})\(R[tz]\eta,\(DR[tz]z\)\eta\).
\end{align*}
Now, since $N\in C^\infty(\C^2,\C^2)$, for $w\in B_{\C^2}(0,1)$ we have $\|D^jN(w)\|_{\mathcal{L}^j(\C^2,\C^2)}\lesssim 1$.
Thus, we have
\begin{align*}
\|DN(u_{t,\tau})\|_{\mathcal{L}(l^{2,s})}+\|DN(v_{t,\tau})\|_{\mathcal{L}(l^{2,s})}\lesssim_s 1,
\end{align*}
and
\begin{align*}
\|D^2N(x_{t,\tau})(\psi_1,\psi_2)\|_{l^{2,s}} &\lesssim_s \|\psi_1\|_{l^{2,2s}} \|\psi_2\|_{l^{2,-s}},\\
\|D^3N(x_{t,\tau})(\psi_1,\psi_2,\psi_3)\|_{l^{2,s}}&\lesssim_s \|\psi_1\|_{l^{2,3s}}\|\|\psi_2\|_{l^{2,-s}}\|\psi_3\|_{l^{2,-s}},
\end{align*}
for $x=u,v$.
Further, by $\|U\|_{\mathcal{L}(l^{2,s})}\lesssim_s 1$ we have, $\|u_{t,\tau}\|_{l^{2,\tau}}+\|v_{t,\tau}\|_{l^{2,s}}\lesssim_s |z|+\|\eta\|_{l^{2,s}}$ and by Lemma \ref{lem:modcoor}, we have $\|R[tz]\eta\|_{l^{2,s}}\lesssim_s\|\eta\|_{l^{2,s}}$, $\|(DR[tz]z)\eta\|_{l^{2,3s}}\lesssim_s |z|^2\|\eta\|_{l^{2,-s}}$ and $\|w_{t,\tau}\|_{l^{2,3s}}\lesssim_s |z|+\|\eta\|_{l^{2,-s}}$.
Taking all these estimates into account, we have
\begin{align}\label{eq:F32est}
\|F_{32}[z,\eta]\|_{l^{2,s}}\lesssim_s (|z|+\|\eta\|_{l^{2,-s}})\|\eta\|_{l^{2,-s}}^2.
\end{align}
Thus, by Proposition \ref{prop:inhomest}, we have
\begin{align*}
\|\sum_{\tau=0}^{t}U^{2(t-\tau)}P_c F_{32}[z(s),\eta(s)]\|_{\mathrm{Stz}_{T+1}\cap l^2_{T+1}l^{2,s}}\lesssim_s \(\|z\|_{l^\infty_T}+\|\eta\|_{\mathrm{Stz}_T}\)\|\eta\|_{\mathrm{Stz}_T}\|\eta\|_{l^2_Tl^{2,-s}}.
\end{align*}
Therefore, we have the conclusion.
\end{proof}

\begin{lemma}\label{lem:boot2}
We have
\begin{align*}
\|\sum_{\tau=0}^{t-1}U^{2(t-\tau)}P_c  F_1[z(\tau),z(\tau+1)]
\|_{\mathrm{Stz}_{T+1}\cap l^2_{T+1}l^{2,-2}}\lesssim \|Z\|_{l^2_T}.
\end{align*}
Recall, $Z$ is defined in \eqref{eq:defZ}.
\end{lemma}

\begin{proof}
The statement follows from Proposition \ref{prop:inhomest} and 
\begin{align}\label{eq:F1bound}
\|F_1[z(t),z(t+1)]\|_{l^{2,2}}\lesssim  |Z(t)|,
\end{align}
which follows from Taylor expansion and Proposition \ref{prop:nlbs}.
\end{proof}

\begin{lemma}\label{lem:boot4}
	Under the assumption of \eqref{prop:boost}, we have
\begin{align}\label{eq:discest1}
\|Z\|_{l^1_T}\lesssim \(\epsilon^{6}C_1^7+\epsilon^2 C_1^2 \) \|\eta(0)\|_{l^2}.
\end{align}
\end{lemma}

\begin{proof}
By \eqref{eq:z2}, \eqref{eq:defF2F3} and \eqref{eq:defF31F32}, we have
\begin{align*}
|Z(t)|^2\lesssim& \(\|F_{31}[\eta(t)]\|_{l^{2,-s}}+\|F_{32}[z(t),\eta(t)]\|_{l^{2,-s}}\)|Z(t) | +|g(t)|+\|\eta\|_{l^{2,-s}}|Z(t)|^2.
\end{align*}
Further, by \eqref{eq:F2}, we have
\begin{align*}
|g(t)|\lesssim \(|z(t)|^2+|z(t+1)|^2\)|Z(t)|^2\lesssim \epsilon^2|Z(t)|^2.
\end{align*}
Thus, we have
\begin{align*}
|Z(t)|^2&\lesssim \|F_{31}[\eta(t)]\|_{l^{2,-s}}^2+\|F_{32}[z(t),\eta(t)]\|_{l^{2,-s}}^2  \lesssim \|\eta(t)\|_{l^{14}}^{14} + \epsilon^2\|\eta\|_{l^{2,-s}}^4.
\end{align*}
Thus, we have \eqref{eq:discest1}.
\end{proof}

\begin{proof}[Proof of Proposition \ref{prop:boost}]
	By \eqref{eq:eta}, we have
\begin{align*}
	&\eta(t)=U^{2t}\eta(0)\\&+\sum_{s=0}^{t-1}U^{2(t-s)}P_c\(\(L[z(s)]-U^2\)R[z(s)]\eta(s)+G[z,R[z]\eta](s)+e^{\im \lambda(z(s))}\Phi_+[z(s)]-\Phi_+[z(s+1)]\).
\end{align*}
Thus, by Proposition \ref{prop:homest} and Lemmas \ref{lem:boot1}, \ref{lem:boot2} and \ref{lem:boot3} we have
\begin{align*}
	\|\eta\|_{\mathrm{Stz}_{T+1}\cap l^{2}_{T+1}l^{2,-2}}\lesssim \(1+C_1^7\epsilon^6 + C_1^2\epsilon^2\)\|\eta(0)\|_{l^2} + \|e^{\im \Lambda_+[z]}z-z(\cdot+1)\|_{l^2_T}
\end{align*}
Further, by Lemma \ref{lem:boot4} and $l^1\hookrightarrow l^2$, we have
\begin{align*}
	\|\eta\|_{\mathrm{Stz}_{T+1}\cap l^{2}_{T+1}l^{2,-2}}\leq C\(1+C_1^7\epsilon^6 + C_1^2\epsilon^2 \)\|\eta(0)\|_{l^2},
\end{align*}
where the constant $C$ is independent of $\epsilon$ and $C_1$.
Thus, taking $C_0=2C$ and for $C_1>C_0$ taking $\epsilon_1>0$ sufficiently small so that if $\epsilon\in (0,\epsilon_1)$, then
\begin{align*}
	C\(1+C_1^7\epsilon^6 + C_1^2\epsilon^2\)<C_1.
\end{align*}
Thus, we have the conclusion.
\end{proof}

In the following we take $C_1=C_0$ and $\epsilon_1$ being the small constant determined by Proposition \ref{prop:boost}.

\begin{corollary}\label{cor:main}
	If $\|u_0\|_{l^2}<\epsilon_1$, then
	\begin{align}
		\|\eta\|_{\mathrm{Stz}\cap l^{2}l^{2,-2}} &\lesssim \|\eta(0)\|_{l^2},\label{eq:spest1}\\
		\|Z\|_{l^1} &\lesssim \epsilon^2  \|\eta(0)\|_{l^2},\label{eq:discest2}\\
\|F[z,z(\cdot+1),\eta]\|_{\mathrm{Stz}^*+l^2l^{2,2}}&\lesssim \epsilon^2\|\eta(0)\|_{l^2}.\label{eq:adjointest}
	\end{align}
\end{corollary}

\begin{proof}
	\eqref{eq:spest1} is a direct consequence of induction and Proposition \ref{prop:boost}.
	\eqref{eq:discest2} follows from \eqref{eq:spest1} and Lemma \ref{lem:boot4}.
For \eqref{eq:adjointest}, we estimate $F_1,F_2$ and $F_3$ separately.
First, for $F_1$, by \eqref{eq:F1bound} and we have \eqref{eq:discest2}, we have
\begin{align*}
\|F_1[z,z(\cdot+1)]\|_{l^1l^2}\lesssim \|Z\|_{l^1}\lesssim \epsilon^2 \|\eta(0)\|_{l^2}.
\end{align*}
For $F_2$, by \eqref{eq:aprioribound}, \eqref{eq:F2bound} and \eqref{eq:spest1}, we have
\begin{align*}
\|F_2[z,\eta]\|_{l^2l^{2,2}}\lesssim \epsilon^2\|\eta(0)\|_{l^2},
\end{align*}
Finally, for $F_3=F_{31}+F_{32}$, from \eqref{eq:F31est} and \eqref{eq:spest1}, we have
\begin{align*}
\|F_{31}[\eta]\|_{\mathrm{Stz}^*}\lesssim \|\eta(0)\|_{l^2}^7\lesssim\epsilon^2\|\eta(0)\|_{l^2},
\end{align*}
and by \eqref{eq:F32est} and \eqref{eq:spest1}, we have
\begin{align*}
\|F_{32}[z,\eta]\|_{l^2l^{2,2}}\lesssim \epsilon \|\eta(0)\|_{l^2}^2\lesssim \epsilon^2\|\eta(0)\|_{l^2}.
\end{align*}
Therefore, we have the conclusion.
\end{proof}

\begin{proof}[Proof of Theorem \ref{thm:main}]
	Set $f(s)=|e^{\im \Lambda_+[z(t)]}z(t) +s\(z(t+1)-e^{\im \Lambda_+[z(t)]}z(t)\)|^2$.
	Then, we have $f(0)=|z(t)|^2$ and $f(1)=|z(t+1)|^2$.
	By fundamental theorem of calculus, we have
	\begin{align*}
		\left||z(t+1)|^2-|z(t)|^2\right|\lesssim \|z\|_{l^\infty}|Z(t)|
	\end{align*}
	Thus,
	\begin{align*}
	\| |z(\cdot+1)|^2-|z|^2 \|_{l^1}\lesssim \epsilon^3 \|\eta(0)\|_{l^2}.
	\end{align*}
	Thus, we see $\lim_{t\to \infty}|z(t)|^2=:\rho^2$ exists and moreover we have
	\begin{align*}
	\left||z(0)|^2-\rho^2\right|\lesssim \epsilon^3 \|\eta(0)\|_{l^2}.
	\end{align*}
	
Next, we show that there exists $\eta_+\in l^2$ s.t.\ $\|\eta_+\|_{l^2}\lesssim \|P_cu_0\|_{l^2}$ and $\|\xi(t)-U_\infty^{2t}\eta_+\|_{l^2}\to 0$ as $t\to \infty$.
Since $\eta\in l^2l^{2,-2}$, we see that $\eta(t)\to 0$ as $t\to \infty$ in $l^{2,-2}$. Thus, recalling $\xi(t)=R[z(t)]\eta(t)$ and $\|\|R[z(t)]-1\|_{\mathcal{L}(l^{2,-2},l^2)}\lesssim \epsilon^2$, we have
	\begin{align*}
	\|\xi(t)-\eta(t)\|_{l^2}\to 0,\ t\to \infty.
	\end{align*} 
Further, since $\|P_cu_0\|_{l^2}\sim \|\eta(0)\|_{l^2}$, it suffices to show 
$\|\eta_+\|_{l^2}\lesssim \|\eta(0)\|_{l^2}$ and $\|\eta(t)-U_\infty^{2t}\eta_+\|_{l^2}\to 0$ as $t\to \infty$.
However, by \eqref{eq:dualest} of Proposition \ref{prop:homest}, we have
\begin{align*}
\|U^{-2t_1}\eta(t_1)-U^{-2t_2}\eta(t_2)\|_{l^2}\lesssim \|F\|_{\mathrm{Stz}^*([t_1,t_2-1])+l^2([t_1,t_2-1],l^{2,2})}.
\end{align*}
Since by \eqref{eq:adjointest} of Corollary \ref{cor:main}, we have
\begin{align*}
\|F\|_{\mathrm{Stz}^*+l^2l^{2,2}}\lesssim \epsilon^2\|\eta(0)\|_{l^2}<\infty,
\end{align*}
we see $U^{-2t}\eta \to \eta_1$ in $l^2$ for some $\eta_1$ satisfying $\|\eta_1\|_{l^2}\sim \|\eta(0)\|_{l^2}\sim \|P_cu(0)\|_{l^2}$.
Finally, by the completeness of the wave operator, there exists $\eta_+$ s.t.\ $\|\eta_+\|_{l^2}\sim \|\eta_1\|_{l^2}$ and $\lim_{t\to\infty}\|U^{2t}\eta_1-U_\infty^{2t}\eta_+\|_{l^2}=0$.
Thus, we have the conclusion.
\end{proof}

Finally, we prove Corollary \ref{cor:orb}.
\begin{proof}[Proof of Corollary \ref{cor:orb}]
Fix $z_0\in\C$ sufficiently small so that we have $\|\Phi_+[z_0]\|_{l^2}<\epsilon_1/2$.
Without loss of generality, we can assume $z_0>0$.
Let $\epsilon<\epsilon_1/2$ and take $u_0\in l^2$ s.t.\ $\|u_0-\Phi_+[z_0]\|_{l^2}<\epsilon$.
Let $u(t)$ be the solution of nonlinear QW \eqref{eq:nlqw2} with $u(0)=u_0$ and set $z(t),\eta(t)$ be the coordinate of $u(t)$ given by \eqref{eq:modcoor}.

First, applying $P_c$ to $u_0$, by Proposition \ref{prop:nlbs} and Lemma \ref{lem:modcoor}, we have $P_cu_0=O(|z(0)|^3)+\eta(0)$ and $(u_0,\phi)=z(0)+O(|z(0)|^2\|\eta(0)\|_{l^2})$.
Thus,
\begin{align}
\epsilon \gtrsim |(u_0-\Phi_+[z_0],\phi_0)|+\|P_c\(u_0-\Phi_+[z_0]\)\|_{l^2},
\end{align}
implies $|z(0)-z_0|+\|\eta(0)\|_{l^2}\lesssim \epsilon$.
By \eqref{eq:spest1} and \eqref{eq:discest2}, we have
\begin{align*}
\sup_{t>0}\(\|\eta(t)\|_{l^2}+\left||z(t)|-|z(0)|\right|\)\lesssim \epsilon.
\end{align*}
Thus, taking $\theta(t)$ so that $e^{-\im \theta(t)}z(t)>0$, we have
\begin{align}
\|u(t)-e^{\im \theta(t)}\Phi[z_0]\|_{l^2}\sim \|\eta(t)\|_{l^2}+ \left||z(t)|-|z(0)|\right|.
\end{align}
Therefore, we have the conclusion.
\end{proof}

\appendix
\section{Proof of Proposition \ref{prop:ede}}\label{sec:A1}

The proof of Proposition \ref{prop:ede} is parallel to the Schr\"odinger operator case.
We follow section 15.5 of \cite{SimonComplexAnalysis2}.
First, recall (see, e.g. \cite{MSSSSdis}) that $U\phi = e^{\im \lambda}\phi$ can be rewritten as
\begin{align}\label{eq.app.3}
	\psi(\cdot+1) = T_\lambda \psi,
\end{align}
where
\begin{align*}
	\psi(x)=\begin{pmatrix}
		\phi_{\downarrow}(x-1)\\ \phi_{\uparrow}(x)
	\end{pmatrix},\quad T_\lambda(x)=\frac{1}{\sqrt{1-|\alpha(x)|^2}}\begin{pmatrix}
	e^{\im(\lambda-\theta(x))} & \alpha(x)\\ \overline{\alpha(x)} & e^{-\im (\lambda-\theta(x))}
\end{pmatrix}.
\end{align*}
Since the argument will be symmetric, it suffices to prove 
\begin{align}\label{eq.app.4}
	\|\psi(x)\|_{\C^2}\sim e^{-\xi x} \ \text{for }x>0,
\end{align}
where  $\xi=\xi(\lambda)$ is given in Proposition \ref{prop:ede}.
We set $T_{\lambda,\infty}=\lim_{x\to \infty}T_{\lambda}(x)$ and 
\begin{align*}
	\widetilde{T}(x)=\begin{pmatrix}
		t_{11}(x) & t_{12}(x)\\ t_{21}(x) & t_{22}(x)
	\end{pmatrix}:=T_\lambda(x)-T_{\lambda,\infty}.
\end{align*}
It is easy to show $\sum_{j,k=1}^2|t_{jk}(x)|\sim \|C(x)-C_\infty\|_{\mathcal{L}(\C^2)}$.
Thus, for arbitrary $\epsilon>0$, there exists $x_\epsilon>0$ s.t. 
\begin{align}\label{Eq.app.1}
	\sum_{j,k=1}^2\|t_{jk}\|_{l^1(\Z_{\geq x_{\epsilon}})}<\epsilon,\ \text{where}\ \Z_{\geq x_{\epsilon}}:=[x_{\epsilon},\infty)\cap \Z.
\end{align}
We further set
\begin{align*}
	\varphi_{\pm}:=\begin{pmatrix}
		\alpha_\infty\\ \gamma_\pm
	\end{pmatrix},\ \varphi_{\pm}(x)=e^{\pm \xi x}\varphi_{\pm}\ \text{and }
\Phi(x) =\(\varphi_+(x)\ \varphi_-(x)\),
\end{align*}
where $\gamma_\pm:=e^{\pm \xi }-e^{\im \lambda}/\sqrt{1-|\alpha_\infty|^2}$.
Then, $\varphi_{\pm}$ are the eigenvectors of $T_{\lambda,\infty}$ associated with $e^{\pm \xi }$ and we have
\begin{align*}
	\Phi(x+1)=T_{\lambda,\infty}\Phi(x).
\end{align*}
We have $\det \Phi(x)=-2\alpha_\infty\sinh \xi \neq 0$ (recall Remark \ref{rem:disc}), so $\Phi(x)$ is invertible and
\begin{align*}
	\Phi(x)^{-1}=-\frac{1}{2\alpha_\infty \sinh \xi }\begin{pmatrix}
		e^{-\xi x}\gamma_- & -e^{-\xi x}\alpha \\ -e^{\xi x}\gamma_+ & e^{\xi }\alpha
	\end{pmatrix}.
\end{align*}
Now, setting $w(x)=\Phi(x)^{-1}\psi(x)$, we have
\begin{align}\label{eq.app.2}
w(x+1)=w(x)+V(x)w(x),\  V(x):=\Phi(x+1)^{-1}\widetilde{T}(x)\Phi(x)=\begin{pmatrix} v_{11}(x) & e^{-2\xi x}v_{12}(x) \\ e^{2\xi x}v_{21}(x) & v_{22}(x)
	\end{pmatrix}.
\end{align}
where
\begin{align*}
	v_{11}(x)&=e^{-\xi}\(\alpha \gamma_- t_{11}+\gamma_+\gamma_- t_{12}-\alpha^2 t_{21}-\alpha \gamma_+ t_{22}\),\\
	v_{12}(x)&=e^{-\xi}\(\alpha \gamma_- t_{11}+\gamma_-^2t_{11}-\alpha^2 t_{21}-\alpha \gamma_- t_{22}\),\\
	v_{21}(x)&=e^{\xi}\(-\alpha \gamma_+ t_{11}-\gamma_+^2t_{12}+\alpha^2t_{21}+\alpha \gamma_+ t_{22}\),\\
	v_{22}(x)&=e^{\xi}\(-\alpha \gamma_+ t_{11}-\gamma_+\gamma_- t_{12}+\alpha^2 t_{21}+\alpha \gamma_- t_{22}\).
\end{align*}
Since $v_{jk}$ are linear combination of $t_{jk}$, by \eqref{Eq.app.1} taking $x_\epsilon$ larger if necessary, we have for arbitrary $\epsilon>0$,
\begin{align*}
	\sum_{j,k=1}^2\|v_{jk}\|_{l^1(\Z_{\geq x_\epsilon})}<\epsilon.
\end{align*}

Now, set
\begin{align*}
X_{x_0}:=\{w:\Z_{\geq x_0}\to \C^2\ |\ \|w\|_X:=\sup_{x\geq x_0}\(|e^{2\xi x}w_{\uparrow}(x)|+|w_{\downarrow}(x)|\)<\infty,\ \lim_{x\to \infty}|e^{2\xi x}w_{\uparrow}(x)|=0\}.
\end{align*}
Then, $X_{x_0}$ is a Banach space.
We further set 
\begin{align*}
	\Psi(w)(x):=\begin{pmatrix}
		0\\ 1
	\end{pmatrix}-\sum_{y=x}^{\infty}V(y)w(y).
\end{align*}
It is straightforwad to show that $\Psi$ is a contraction on $X_{x_0}$ for sufficently large $x_0$.
Thus, there exists a unique fixed point $w_0$ of $\Psi$ on $X_{x_0}$.
The fixed point $w_0 \in X$ satisfies \eqref{eq.app.2} and $\lim_{x\to \infty} w_0(x)=\begin{pmatrix}
	0\\ 1
\end{pmatrix}$.
Now, reverting the procedure and setting $\psi_0:=\Phi(x)w_0(x)$, we see that $\psi_0$ satisfies \eqref{eq.app.3} and \eqref{eq.app.4}.
Finally, the solution of \eqref{eq.app.3} is two dimensional and since $w_0(x)\to 0$ as $x\to \infty$, the other solution cannot be bounded and hence cannot be in $l^2$.
Thus, we have $\psi = c\psi_0$ with some constant $c$.
Therefore, we have the conclusion.

\section{Kato smoothness for unitary operators}\label{sec:A2}

In this section, we introduce a unitary analog of Kato smoothness.
We follow \cite{Kato65MA} and \cite{RS4}.
In the following $U$ will be an unitary operator on a Hilbert space $H$.
As before, we set
\begin{align*}
	R_U(\mu):=(Ue^{-\im \mu}-1)^{-1},\ \mu\in\C/2\pi\Z.
\end{align*}

The main result in this section  is the following:
\begin{theorem}\label{thm.KSU.4.1}
	Let $A\in \mathcal{L}(H)$ and $U$ be a unitary operator on $H$.
	Then, if one of the following quantities are finite, then all the others are finite.
	Moreove, if so, there are all equal.
	\begin{enumerate}
		\item $\|AU^{\cdot}\|_{\mathcal{L}(H,l^2(\Z,H))}^2$.
		\item 
		$\frac{1}{4\pi^2}\sup_{\substack{\epsilon\neq 0\\ \|\varphi\|_H=1}}\(\|AR_U(\cdot+\im \epsilon)\varphi\|_{L^2(\T,H)}^2+\|AR_U(\cdot-\im \epsilon)\varphi\|_{L^2(\T,H)}^2\)$.
		\item 
		$\frac{1}{2\pi}\sup_{\Im \mu\neq 0}\|A\(R_U(\mu)-R_U(\overline{\mu})\)A^*\|_{\mathcal{L}(H)}$
		\item $\sup_{[a,b)\subset \T}\frac{\|A1_{[a,b)}(U)\|_{\mathcal{L}(H)}^2}{|b-a|}$
	\end{enumerate}
\end{theorem}

\begin{definition}\label{def:KS}
	If one (hence all) of the quantities 1.$\sim$4. of Theorem \ref{thm.KSU.4.1} is finite, we say $A$ is $U$-smooth.
\end{definition}

Before going in to the proof of Theorem \ref{thm.KSU.4.1}, we give an sufficient condition of $U$-smoothness.

\begin{corollary}\label{cor:KS}
	If $\sup_{\Im \mu\neq 0}\|AR_U(\mu)A^*\|_{\mathcal{L}(H)}<\infty$, then $A$ is $U$ smooth.
\end{corollary}

\begin{proof}
	From triangle inequality we can bound 3.\ of Theorem \ref{thm.KSU.4.1} by $\sup_{\Im \mu\neq 0}\|AR_U(\mu)A^*\|_{\mathcal{L}(H)}$.
\end{proof}

To prove theorem \ref{thm.KSU.4.1}, we prepare several formula.
First, by direct computation we can check that for $\lambda,\epsilon\in\R$, we have
	\begin{align}\label{KSU.4.4}
		R_U(\lambda+\im \epsilon)-R_U(\lambda-\im \epsilon)=(1-e^{-2\epsilon})R_U(\lambda-\im \epsilon)^* R_U(\lambda-\im \epsilon).
	\end{align}
	In particular, if $\epsilon>0$, then 
	\begin{align}
		R(\lambda+\im \epsilon)-R(\lambda-\im \epsilon)\geq 0.\label{eq.KSU.4.10}
	\end{align}


We next prepare several integral formula.
In the following $\oint_{|z|=1}$ will mean the integral in the anti-clockwise direction.
\begin{lemma}\label{lem.KSU.4.2}
	Let $\epsilon>0$.
	Then, we have
	\begin{align}\label{KSU.4.2}
		\frac{1}{2\pi}\int_{\T}\(\frac{1}{e^{-\im \lambda}e^{\epsilon}-1}-\frac{1}{e^{-\im \lambda}e^{-\epsilon}-1}\)\,d\lambda=1.
	\end{align}
	Moreover, let $I=(a,b)\subsetneq \T$ and define
	\begin{align*}
		f_{I,\epsilon}(\zeta):=\frac{1}{2\pi}\int_I \(\frac{1}{e^{\im (\zeta-\lambda)}e^{\epsilon}-1}-\frac{1}{e^{\im (\zeta-\lambda)}e^{-\epsilon}-1}\)\,d\lambda,
	\end{align*}
	then $\sup_{\epsilon}\|f_{I,\epsilon}\|_{L^\infty}\leq 1$ and
	\begin{align}\label{KSU.4.3}
		\lim_{\epsilon\to 0+}f_{I,\epsilon}(\zeta)=\begin{cases}
			1 & \zeta\in (a,b),\\
			\frac{1}{2} & \zeta=a,b,\\
			0 & \zeta\not\in (a,b).
		\end{cases}
	\end{align}
\end{lemma} 

\begin{proof}
	Set $z=e^{\im \lambda}$.
	Then, we have $d\lambda = \frac{dz}{\im z}$.
	Thus, by residue theorem we have
	\begin{align*}
		\frac{1}{2\pi}\int_{\T}\(\frac{1}{e^{-\im \lambda}e^{\epsilon}-1}-\frac{1}{e^{-\im \lambda}e^{-\epsilon}-1}\)\,d\lambda&=\frac{1}{2\pi \im}\oint_{|z|=1}\(\frac{1}{e^{\epsilon}-z}+\frac{1}{z-e^{-\epsilon}}\)\,dz\\&
		=1.
	\end{align*}
	For \eqref{KSU.4.3}, first if $\zeta\not\in [a,b]$, then the integrand is uniformly bounded and uniformly converges to $0$.
	Thus we have this case.
	If $\zeta\in (a,b)$ we have
	\begin{align*}
		\lim_{\epsilon\to 0+}f_{I,\epsilon}(\zeta)=\lim_{\epsilon\to 0+}f_{\T,\epsilon}(\zeta)+\lim_{\epsilon\to 0+}f_{\T\setminus[a,b],\epsilon}(\zeta)=1,
	\end{align*}
	by \eqref{KSU.4.2} and the fact $\zeta\not\in \T\setminus[a,b]$.
	
	Next, using the formula \eqref{KSU.4.4} for $U=e^{\im \zeta}$, we see that the integrand is nonnegative.
	Thus,
	\begin{align*}
		|f_{I,\epsilon}(\zeta)|\leq |f_{\T,\epsilon}(\zeta)|=1.
	\end{align*}
	
	Finally, if $\zeta=a$, by proceeding as above and using the formula \eqref{KSU.4.4} with $U=1$, we have
	\begin{align*}
		\lim_{\epsilon\to 0}f_{I,\epsilon}&=\frac{1}{2\pi}\lim_{\epsilon\to 0+}\int_0^\pi \(\frac{1}{e^{-\im \lambda}e^{\epsilon}-1}-\frac{1}{e^{-\im \lambda}e^{-\epsilon}-1}\)\,d\lambda\\&
		=\frac{1}{2\pi}\lim_{\epsilon\to 0+}(1-e^{-2\epsilon})\int_0^\pi \left|\frac{1}{e^{-\im \lambda}e^{-\epsilon}-1} \right|^2\,d\lambda\\&
		=\frac{1}{2\pi}\lim_{\epsilon\to 0+}(1-e^{-2\epsilon})\int_0^\pi \frac{1}{1+e^{-2\epsilon}-2e^{-\epsilon}\cos \lambda} \,d\lambda.
	\end{align*}
	Since, the integrand is even function, we have the conclusion.
\end{proof}

From spectral theorem and Lemma \ref{lem.KSU.4.2}, we have the Stone's formula for unitary operators.
That is, we have
\begin{align}\label{eq.Stone}
	\frac{1}{2}\(1_{(a,b)}(U)+1_{[a,b]}(U)\)=\mathrm{s-lim}_{\epsilon\to 0+}\frac{1}{2\pi}\int_a^b \(R(\lambda+\im \epsilon)-R(\lambda-\im \epsilon)\)\,d\lambda.
\end{align}

%

Similarly, using the formula \eqref{KSU.4.2}, we have
\begin{align}\label{eq:KSU.1}
	\frac{1}{2\pi}\int_{\T}\(\varphi,\(R(\lambda+\im \epsilon)-R(\lambda-\im\epsilon)\)\varphi\)\,d\lambda=\|\varphi\|_H^2.
\end{align}

%
%

Recall we are setting the Fourier transform as
$
	\mathcal{F}u(\lambda):=\frac{1}{2\pi}\sum_{t\in\Z}e^{-\im \lambda t }u(t).
$
Then, we have Plancherel theorem
$
	\|u\|_{l^2(\Z,H)}=\|\mathcal{F}u\|_{L^2(\T,H)}
$
\begin{lemma}
	Let $\epsilon>0$.
	Then, we have
	\begin{align}
		\mathcal{F}\(e^{-\epsilon(\cdot)} 1_{[0,\infty)}(\cdot)A U^\cdot \varphi \)(\lambda)&=-\frac{1}{2\pi}A R(\lambda-\im \epsilon) \varphi,\label{eq.KSU.4.5}\\
		\mathcal{F}\(e^{\epsilon(\cdot )} 1_{(-\infty,-1]}(\cdot)A U^\cdot \varphi \)(\lambda)&=-\frac{1}{2\pi}A R(\lambda+\im \epsilon) \varphi.\label{eq.KSU.4.6}
	\end{align}
	In particular, we have
	\begin{align}
		\mathcal{F}(e^{-\epsilon|\cdot|}AU^\cdot \varphi)(\lambda)=\frac{1}{2\pi}A\(R(\lambda+\im \epsilon)-R(\lambda-\im \epsilon)\)\varphi.\label{eq.KSU.4.7}
	\end{align}
\end{lemma}

\begin{proof}
	For \eqref{eq.KSU.4.5},
	\begin{align*}
		\mathcal{F}\(e^{-\epsilon(\cdot-T)} 1_{[T,\infty)}(\cdot)A U^\cdot \varphi \)(\lambda)&=\frac{1}{2\pi}A\sum_{t=T}^{\infty}e^{-\im \lambda t}e^{-\epsilon(t-T)}U^t\varphi\\&
		=e^{-\im \lambda T}\frac{1}{2\pi}A(1-e^{-\im(\lambda-\im \epsilon)}U)^{-1}U^T\varphi\\&
		=-e^{-\im \lambda T}\frac{1}{2\pi}AR(\lambda-\im \epsilon)U^T\varphi.
	\end{align*}
	For \eqref{eq.KSU.4.6},
	\begin{align*}
		\mathcal{F}\(e^{\epsilon(\cdot-T)} 1_{(-\infty,T-1]}(\cdot)A U^\cdot \varphi \)(\lambda)&=\frac{1}{2\pi}A\sum_{t=-\infty}^{T-1}e^{-\im \lambda t}e^{\epsilon(t-T)}U^t\varphi\\&
		=\frac{1}{2\pi}A\sum_{s=0}^{\infty}e^{-\im \lambda (-s+T-1)}e^{\epsilon(-s-1)}U^{-s+T-1}\varphi\\&
		=e^{-\im \lambda T} e^{\im(\lambda+\im \epsilon)}\frac{1}{2\pi}A\sum_{s=0}^{\infty}e^{\im \lambda s}e^{-\epsilon s}U^{-s}U^{-1}U^T\varphi\\&
		=e^{-\im \lambda T} e^{\im(\lambda+\im \epsilon)}\frac{1}{2\pi}A(1-e^{\im (\lambda+\im \epsilon)}U^{-1})^{-1}U^{-1}U^T\varphi\\&
		=e^{-\im \lambda T} \frac{1}{2\pi}AR(\lambda+\im \epsilon)U^T\varphi.
	\end{align*}
	Finally, we have \eqref{eq.KSU.4.7} by adding \eqref{eq.KSU.4.5} and \eqref{eq.KSU.4.6} with $T=0$.
\end{proof}

\begin{lemma}\label{lem.KSU.A}
	We have
	\begin{align*}
		\|AU^{\cdot}\|_{\mathcal{L}(H,l^2(\Z,H))}^2
		=\frac{1}{4\pi^2}\sup_{\substack{\epsilon\neq 0\\ \|\varphi\|_H=1}} \|AR(\cdot+\im \epsilon)\varphi\|_{L^2(\T,H)}^2+\|AR(\cdot-\im \epsilon)\varphi\|_{L^2(\T,H)}^2.
	\end{align*}
\end{lemma}

\begin{proof}
	First, by Plancherel and \eqref{eq.KSU.4.5}, for $\varphi\in H$ with $\|\varphi\|_H=1$, we have
	\begin{align}
		\sum_{t=0}^\infty\|e^{-\epsilon t}AU^t \varphi\|_{H}^2&=\frac{1}{4\pi^2}\|AR(\lambda-\im \epsilon)\varphi\|_{L^2(\T,H)}^2,\label{eq.KSU.4.8}\\
		\sum_{t=-\infty}^{-1}\|e^{\epsilon t}AU^t \varphi\|_{H}^2&=\frac{1}{4\pi^2}\|AR(\lambda-\im \epsilon) \varphi\|_{L^2(\T,H)}^2. \label{eq.KSU.4.9}
	\end{align}
%
%
%
%
	Summing \eqref{eq.KSU.4.8} and \eqref{eq.KSU.4.9} we have
	\begin{align*}
		\|e^{-\epsilon|t|}AU^t\varphi\|_{l^2(\Z,H)}^2=\frac{1}{4\pi^2}\(\|AR(\lambda+\im \epsilon)\varphi\|_{L^2(\T,H)}^2+\|AR(\lambda+\im \epsilon)\varphi\|_{L^2(\T,H)}^2\).
	\end{align*}
Therefore, we have the conclusion.
\end{proof}

\begin{lemma}\label{lem.KSU.B}
	We have
	\begin{align*}
		\|AU^{\cdot}\|_{l^2(\Z,H)}^2&=\frac{1}{2\pi}\sup_{\Im \mu\neq 0}\|A\(R(\mu)-R(\overline{\mu})\)A^*\|_{\mathcal{L}(H)}
		=\sup_{I\subset \T}\frac{\|AE_I\|_{\mathcal{L}(H)}^2}{|I|}.
	\end{align*}
\end{lemma}

\begin{proof}
	First, by Plancherel and \eqref{eq.KSU.4.7},
	\begin{align}
		\|e^{-\epsilon|t|}AU^t \varphi\|_{l^2(\Z,H)}^2=\frac{1}{4\pi^2}\|A(R(\lambda+\im \epsilon)-R(\lambda-\im \epsilon))\varphi\|_{L^2(\T,H)}^2\label{eq.KSU.4.12}
	\end{align}
	From \eqref{eq.KSU.4.10}, there exists a self-adjoint $K(\lambda+\im \epsilon)$ s.t. 
	\begin{align*}
		R(\lambda+\im \epsilon)-R(\lambda-\im \epsilon)=K(\lambda+\im \epsilon)^2.
	\end{align*}
	Thus,
	\begin{align*}
		\|e^{-\epsilon|t|}AU^t \varphi\|_{l^2(\Z,H)}^2&=\frac{1}{4\pi^2}\|AK(\lambda+\im \epsilon)^2\varphi\|_{L^2(\T,H)}^2\\&\leq \frac{1}{4\pi^2} \|AK(\lambda+\im \epsilon)\|_{L^\infty(\T,\mathcal{L}(H))}^2\|K(\lambda+\im \epsilon)\varphi\|_{L^2(\T,H)}^2
	\end{align*}
	Now, by \eqref{eq:KSU.1}, we have
	\begin{align*}
		\|K(\lambda+\im \epsilon)\varphi\|_{L^2(\T,H)}^2=\int_{\T}\(\varphi,\(R(\lambda+\im \epsilon)-R(\lambda-\im \epsilon)\)\varphi\)\,d\lambda=2\pi \|\varphi\|_{H}^2.
	\end{align*}
	On the other hand,
	\begin{align*}
		\|AK(\lambda+\im \epsilon)\|_{L^\infty(\T,\mathcal{L}(H))}^2&=\|K(\lambda+\im \epsilon)A^*\|_{L^\infty(\T,\mathcal{L}(H))}^2=\sup_{\lambda\in\T}\sup_{\|\varphi\|_H=1}\|K(\lambda+\im \epsilon)A^*\varphi\|_{H}^2\\&
		=\sup_{\lambda\in\T}\sup_{\|\varphi\|_H=1}\left|\(\varphi,A\(R(\lambda+\im \epsilon)-R(\lambda-\im \epsilon)\)A^*\varphi \)\right|\\&
		=\sup_{\lambda\in\T}\|A(R(\lambda+\im \epsilon)-R(\lambda-\im \epsilon))A^*\|_{\mathcal{L}(H)}.
	\end{align*}
	Here, we used the fact that $A(R(\lambda+\im \epsilon)-R(\lambda-\im \epsilon))A^*$ is self-adjoint, which follows from \eqref{KSU.4.4}.
	Thus, we obtain
	\begin{align*}
		\|e^{-\epsilon|t|}AU^t \varphi\|_{l^2(\Z,H)}^2\leq \frac{1}{2\pi} \sup_{\Im\mu\neq 0}\|A(R(\mu)-R(\overline{\mu}))A^*\|_{\mathcal{L}(H)}\|\varphi\|_H^2.
	\end{align*}
	This implies
	\begin{align*}
		\|AU^{\cdot}\|_{\mathcal{L}(H,l^2(\Z,H))}^2\leq \frac{1}{2\pi}\sup_{\Im\mu\neq 0}\|A(R(\mu)-R(\overline{\mu}))A^*\|_{\mathcal{L}(H)}.
	\end{align*}
	Next, let $\mu$ be the spectral measure for $U$ w.r.t.\ $A^*\varphi$.
	Then, for any interval $[a,b)\subset \T=[0,2\pi)$, we have
	\begin{align*}
		\mu([a,b))= \(A^*\varphi,1_{[a,b)}(U) A^*\varphi\) \leq \(\sup_{[a,b)\subset \T}\frac{\|A1_{[a,b)}(U)A^*\|_{\mathcal{L}(H)}}{|b-a|}\)\|\varphi\|_H^2 |b-a| .
	\end{align*}
	Thus, $\mu$ is absolute continuous.
	Further, there exists $g$ s.t. $d\mu=g d\zeta$ with
	\begin{align}
		\|g\|_{L^\infty(\T)}\leq \(\sup_{[a,b)\subset \T}\frac{\|A1_{[a,b)}(U)A^*\|_{\mathcal{L}(H)}}{|b-a|}\)\|\varphi\|_H^2=\(\sup_{[a,b)\subset \T}\frac{\|A1_{[a,b)}(U)\|_{\mathcal{L}(H)}^2}{|b-a|}\)\|\varphi\|_H^2.\label{App.1}
	\end{align}
	Thus, by Lemma \ref{lem.KSU.4.2} (if there exists an eigenfunction $\varphi$ of $U$ s.t. $\varphi\not\in \mathrm{Ker}A $, then r.h.s.\ of \eqref{App.1} is infinite so without loss of generality we can assume $A(1_{(a,b)}(U)+1_{[a,b]}(U))=A1_{[a,b)}(U)$) we have
	\begin{align}
		&\frac{1}{2\pi}\(\varphi,A(R(\mu)-R(\bar{\mu}))A^*\varphi\)=\frac{1}{2\pi}\left|\int_{\T}\(\frac{1}{e^{\im (\zeta-\lambda)}e^{\epsilon}-1}-\frac{1}{e^{\im (\zeta-\lambda)}e^{-\epsilon}-1}\)g(\zeta)\,d\zeta\right|\nonumber\\&
		\leq \frac{1}{2\pi} \int_{\T}\(\frac{1}{e^{\im (\zeta-\lambda)}e^{\epsilon}-1}-\frac{1}{e^{\im (\zeta-\lambda)}e^{-\epsilon}-1}\) \,d\zeta  \sup_{[a,b)\subset \T}\frac{\|A1_{[a,b)}(U)\|_{\mathcal{L}(H)}^2}{|b-a|}\|\varphi\|_H^2 \label{eq.KSU.4.11} \\&
		= \sup_{[a,b)\subset \T}\frac{\|A1_{[a,b)}(U)\|_{\mathcal{L}(H)}}{|b-a|}\|\varphi\|_H^2.\nonumber
	\end{align}
	Notice that the integrand in \eqref{eq.KSU.4.11} is positive from \eqref{eq.KSU.4.10} with $U=e^{\im \zeta}$.
	Thus, we have
	\begin{align*}
		\frac{1}{2\pi}\sup_{\Im\mu\neq 0}\|A(R(\mu)-R(\overline{\mu}))A^*\|_{\mathcal{L}(H)}\leq 
		\sup_{[a,b)\subset \T}\frac{\|A1_{[a,b)} \|_{\mathcal{L}(H)}^2}{|b-a|}.
	\end{align*}
	Finally, from Stone's formula \eqref{eq.Stone} we have
	\begin{align*}
		\frac{\|A1_{[a,b)}(U) \|_{\mathcal{L}(H)}^2}{|b-a|}&=\frac{1}{|b-a|}\sup_{\|\varphi\|_H,\|\psi\|_H = 1}\(A^*\psi,1_{[a,b)}(U)\varphi\)^2\\&
		=\frac{1}{4\pi^2}\frac{1}{|b-a|}\sup_{\|\varphi\|_H,\|\psi\|_H = 1}\(\lim_{\epsilon\to 0+}\int_a^b\(A^*\psi,(R(\lambda+\im \epsilon)-R(\lambda-\im \epsilon)\varphi)\)\,d\lambda\)^2\\&\leq
		\frac{1}{4\pi^2}\sup_{\|\varphi\|_H=1}\lim_{\epsilon\to 0+}\|A\(R(\lambda+\im \epsilon)-R(\lambda-\im \epsilon)\)\varphi\|_{L^2(\T,H)}^2\\&=\sup_{\|\varphi\|_H=1}\lim_{\epsilon\to 0+}\|e^{-\epsilon|\cdot|}AU^{\cdot}\varphi\|_{l^2(\Z,H)}^2=\|AU^{\cdot}\varphi\|_{\mathcal{L}(H,l^2(\Z,H))}^2.
	\end{align*}
Therefore, we have the conclusion.
\end{proof}

Combining Lemmas \ref{lem.KSU.A} and \ref{lem.KSU.B}, we have Theorem \ref{thm.KSU.4.1}.

\medskip

Masaya Maeda

Department of Mathematics and Informatics,
Faculty of Science,
Chiba University,
Chiba 263-8522, Japan

{\it E-mail Address}: {\tt maeda@math.s.chiba-u.ac.jp}

\end{document}